%% file: main.tex
\documentclass[11pt]{article}
\usepackage[utf8]{inputenc}
\usepackage{fullpage}
\usepackage{amsmath,amsthm,amssymb, graphicx, enumerate}

\DeclareMathOperator {\spn}  {span}

\newtheorem{theorem}{Theorem}[section]
\newtheorem*{namedtheorem}{\theoremname}
\newcommand{\theoremname}{testing}

\newtheorem{lemma}[theorem]{Lemma}
\newtheorem{claim}[theorem]{Claim}

\newtheorem{fact}[theorem]{Fact}
\newtheorem{corollary}[theorem]{Corollary}
\newtheorem{question}[theorem]{Question}

\theoremstyle{definition}
\newtheorem{definition}[theorem]{Definition}

\newcommand {\Exp}       {\mathbb{E}}

\newcommand {\Prb}  [1] {\Pr \left[#1 \right]}

\newcommand {\E}     [1] {\Exp\left[#1\right]}

\newcommand{\iprod}[1]{\langle #1 \rangle}

\newcommand{\opr}{\otimes}
\newcommand{\errp}{\vartheta}

\newcommand {\Iprod} [2] {\left\langle #1, #2 \right\rangle}

\newcommand {\bbN}    {\mathbb{N}}

\newcommand {\bbR}    {\mathbb{R}}

\newcommand {\calD}   {{\cal{D}}}

\newcommand {\calS}   {{\cal{S}}}

\newcommand {\calN}   {{\cal{N}}}

\newcommand {\calF}   {{\cal{F}}}

\newcommand {\calV}   {{\cal{V}}}
\newcommand {\calW}   {{\cal{W}}}

\newcommand{\norm}[1]{\left\|#1 \right\|}
\newcommand{\nrm}[2]{ \left\|#1 \right\|_{ #2 }}

\newcommand{\eps}{\varepsilon}
\newcommand{\poly}{\text{poly}}

\newcommand{\abs}[1]{\left|#1 \right|}
\newcommand {\R}   {{\mathbb{R}}}
\newcommand{\del}{\delta}

\newcommand{\spc}[2]{{#1}^{(#2)}}

\newcommand{\ub}{\rho}
\newcommand{\nz}{nz}
\newcommand{\kr}{K-rank~}
\newcommand{\krk}[1]{\text{K-rank}_{#1}}
\newcommand{\na}{n_A}
\newcommand{\nb}{n_B}
\newcommand{\nc}{n_C}
\newcommand{\ka}{\tau}
\newcommand{\diag}{\text{diag}}
\newcommand{\tens}[3]{[ #1 ~ #2 ~ #3]}
\newcommand{\lmin}[1]{\lambda_{\min}( #1 )}
\newcommand{\lmax}[1]{\lambda_{\max}( #1 )}
\newcommand{\mmax}{c_{max}}
\newcommand{\Bd}{c_{max}}

\newcommand{\sign}{\text{sgn}}

\newcommand{\minl}{L_{\min}}
\newcommand{\maxl}{L_{\max}}
\newcommand{\up}{\tilde{u}_{\perp}}
\newcommand{\vp}{\tilde{v}_{\perp}}
\newcommand{\tU}{\tilde{U}}
\newcommand{\tV}{\tilde{V}}
\newcommand{\tM}{\tilde{M}}
\newcommand{\tw}{\tilde{w}}
\newcommand{\tm}{\tilde{\mu}}
\newcommand{\tu}{\tilde{u}}
\newcommand{\tv}{\tilde{v}}
\newcommand{\tP}{\tilde{P}}

\title{Uniqueness of Tensor Decompositions with Applications to Polynomial Identifiability}
\author{Aditya Bhaskara\thanks{Email: \textsf{bhaskara@cs.princeton.edu}} \\ {\small EPFL, Switzerland} \and Moses Charikar\thanks{Email: \textsf{moses@cs.princeton.edu}. Supported by NSF awards CCF 0832797, AF 0916218 and a Google research award.}\\ {\small Princeton University} \and Aravindan Vijayaraghavan\thanks{Email: \textsf{aravindv@cs.cmu.edu}. Supported by the Simons Postdoctoral Fellowship.} \\ {\small Carnegie Mellon University}}
\date{}
\begin{document}
\maketitle

\begin{abstract}
\input{abstract}
\end{abstract}
\newpage

\section{Introduction}

\input{intro.tex}

\section{Some preliminaries and our results}\label{sec:prelim}
\input{prelim.tex}

\section{Uniqueness of Tensor Decompositions}
\input{robustkruskal.tex}


\section{Computing Tensor Decompositions}\label{sec:algo}
\input{tensordecomp.tex}

\section{Polynomial Identifiability of Latent Variable and Mixture Models} \label{sec:applications}
\input{applications.tex}

\section{Discussion and Open Problems}
\label{sec:openproblems}
\input{openproblems.tex}

\section{Acknowledgements}

We thank Ravi Kannan for valuable discussions about the algorithmic results in this work,
and Daniel Hsu for helpful pointers to the literature. The third author would also like to thank Siddharth Gopal for some useful pointers about HMM models in speech and image recognition. 

\nocite{*}
\bibliographystyle{alpha}
\bibliography{tensor}

\appendix
\input{basiclemmas.tex}
\end{document}

%% file: abstract.tex
We give a robust version of the celebrated result of Kruskal on the
uniqueness of tensor decompositions: we prove that given a tensor whose
decomposition 
satisfies a robust form of Kruskal's rank condition, 
it is possible to approximately recover the decomposition if the tensor is known up to a sufficiently small (inverse polynomial) error.

Kruskal's theorem has found many applications in proving the {\em
identifiability} of parameters for various latent variable models and mixture models such as Hidden Markov models, topic models etc. 
Our robust version immediately implies identifiability using only polynomially many
samples in many of these settings. This polynomial identifiability is an essential
first step towards efficient learning algorithms for these models.

Recently, algorithms based on tensor decompositions have been used to 
estimate the parameters of various hidden variable models efficiently 
in special cases as long as they satisfy certain ``non-degeneracy''
properties. Our methods give a way to go beyond this non-degeneracy barrier, 
and establish polynomial identifiablity of the parameters under much
milder conditions. Given the importance of Kruskal's theorem in the tensor literature, we expect that this robust version will have several applications beyond the settings we explore in this work.



%% file: intro.tex
Statisticians have long studied the identifiability of probabilistic models \cite{Tei61,Tei67,TC82}, i.e. whether the 
parameters of a model can be learned from data generated by the model.
A central question in unsupervised learning \cite{Gha04} is the efficient computation of such latent model parameters from observed data.
A necessary step towards efficient (polynomial time) learning is to show that the parameters are indeed identifiable
after observing polynomially many samples.
The method of moments approach, pioneered by Pearson \cite{Pea94}, infers model parameters from 
empirical moments such as means, pairwise and other higher order correlations.
In general, very high order moments may be needed for this approach to succeed
and the unreliability of empirical estimates of these moments leads to exponential sample complexity
\cite{MV10,BS10,GLPR12}.

An exciting sequence of recent work \cite{MR06,AHK12, HK12, AGHKT12}
has met with considerable success in cases where the underlying
models satisfy a certain non-degeneracy condition (that we will explain later).
Informally, the condition requires that the dimension ($n$) of the observations is at least as large
as the number of possible values ($R$) for the hidden variable and that certain model parameters
are in general position.
The moments are naturally represented by tensors (high dimensional analogs of matrices) and
low rank decompositions of such tensors can be used to deduce the parameters of the underlying model.
Under suitable non-degeneracy assumptions, the required tensor decompositions can be computed
efficiently using an iterative procedure akin to power iteration for computing matrix eigenvalues.
One focus of our work is developing tensor decomposition techniques that apply in more general
settings where these non-degeneracy assumptions are violated, i.e. $n$ is much smaller than $R$.  
Such settings do arise in many cases of practical interest such as in applications of hidden Markov models to speech recognition and image classification, where
the dimension ($n$) of the feature space is typically much smaller than the number of values ($R$) 
for the hidden variable. 
For instance, the (effective) feature space corresponds to just the low-frequency components in the fourier spectrum in speech, or the local neighborhood of a pixel in images. 
These are typically low dimensional than the number of words or image classes.



In fact, the connection of tensor decompositions to learning probabilistic models has been made earlier
in the algebraic statistics literature. 
In a series of papers, identifiability of several latent variable models was established \cite{AMR09,APRS11,RS12} 
via low rank decomposition of certain moment tensors.
A fundamental result of Kruskal \cite{Kru77} on uniqueness of tensor decompositions plays a
crucial role in ensuring that the model parameters are correctly identified by this procedure.
Note that this assumes access to an infinite number of samples and 
does not give any information on the number of samples needed to learn the model
parameters within specified error bounds. 
Kruskal's theorem by itself is not useful for establishing any such sample complexity bounds
since it only guarantees uniqueness for low rank decompositions of the actual moment tensors.
It does not say anything about the decomposition of empirical moment tensors which are approximations
of these. 
In order to understand how large a sample size is needed,
one would need a {\em robust uniqueness} guarantee of this form: 
if the empirical moment tensor $T'$ is close to the moment tensor $T$, then a low rank 
decomposition of $T'$ is (term by term) close to a low rank decomposition of $T$.

Our main technical contribution in this work is establishing such a robust version of Kruskal's
classic uniqueness theorem for tensor decompositions.
This provides a uniqueness guarantee that is directly applicable for establishing polynomial identifiability in a host of applications \cite{AMR09} where Kruskal's theorem was used to prove identifiability
assuming access to exact moment tensors.
Since polynomially many samples from the distribution (typically) yield
an approximation to these tensors up to $1/\poly(n)$ error, our
robust version of Kruskal's theorem establishes polynomial
identifiability in all such applications.
To the best of our knowledge, no such robust version of Kruskal's theorem is known in
the literature.
Given the importance of this theorem in the tensor literature,
we expect that this robust version will have applications beyond the settings we explore in this work.
Our robust uniqueness theorem is accompanied by new algorithms to find low rank tensor 
decompositions.


\subsection{Tensors and their Decompositions}

A tensor is a multidimensional array -- a generalization of vectors and matrices
e.g. an $n_1 \times n_2 \times n_3$ tensor is a 3-tensor which is an element in 
$R^{n_1 \times n_2 \times n_3}$.
Low rank tensor decompositions (analogs of SVD for matrices) have been studied 
intensively as methods for extracting structure in data.
These originated in work of Hitchcock \cite{Hit27} and Cattell \cite{Cat44}.
They were studied in the 60's and 70's in the psychometrics literature and since
the 80's, in the chemometrics literature.
The notion of tensor rank also plays an important role in algebraic complexity,
and is closely connected to the exponent of matrix multiplication.
More recently, tensor decompositions have found applications in signal processing, numerical linear algebra, computer vision,
numerical analysis, data mining, graph analysis, neuroscience and more.

Carroll and Chang \cite{CC70} introduced CANDECOMP (canonical decomposition) and
independently, Harshman \cite{Har70} introduced PARAFAC (parallel factors).
CANDECOMP/PARAFAC is now referred to as CP decomposition \cite{Kie00}.
It expresses a tensor as a sum of rank-one tensors where each rank-one tensor is
the outer product of column vectors.
The rank of a tensor is the minimum number of terms required for such a decomposition.
While the definition of tensor rank is analogous to that of matrix rank, their properties
are quite different.
In fact, computing the rank of a tensor is NP-hard \cite{Has90} and in fact several other
problems associated with low rank approximation of tensors are NP-hard as well \cite{HL13}.

For matrices, a fundamental result of Eckart and Young \cite{EY36} shows that the best rank-$k$
approximation consists of the leading $k$ terms of the SVD.
This is not the case for CP decomposition of tensors -- the best rank one
approximation may not be a factor in the best rank two approximation.
In fact, the best rank $k$-approximation may not exist.
For example, certain tensors of rank-three can be arbitrarily well approximated
by a sequence of rank-two tensors \cite{Knu97,Paa00,dSL08,Lan}.
In fact, the set of tensors of a certain size that do not have a best rank-$k$
approximation has positive volume \cite{dSL08}.
To overcome this problem, the concept of {\em border rank} was introduced 
and studied in the algebraic complexity community.
This is defined to be the minimum number of rank-one tensors that are sufficient
to approximate the given tensor with arbitrarily small error.
In fact, the complexity of matrix multiplication is exactly captured by the border rank
of the associated tensor \cite{KB09,Lan}.

An important property of higher order tensors is that (under certain conditions)  
their minimum rank decompositions are unique upto trivial scaling and permutation.
This is in contrast to matrix decompositions.
Note that the SVD of a matrix is unique (assuming distinct singular values) only because
we impose additional orthogonality constraints.

A classic result of Kruskal \cite{Kru77} gives a sufficient condition for uniqueness of the CP
decomposition of a 3-tensor.
Suppose that a 3-tensor $T$ has the following decomposition:
\begin{align}
T = \tens{A}{B}{C} \equiv \sum_{r=1}^R A_r \otimes B_r \otimes C_r
\label{CPdecomp}
\end{align}
Let the Kruskal rank or \kr $k_A$ of matrix $A$ (formed by column vectors $A_r$)
be the maximum value of $k$ such that any $k$ columns of $A$
are linearly independent.
$k_B$ and $k_C$ are similarly defined.
Kruskal's result says that a sufficient condition for the uniqueness of the decomposition (\ref{CPdecomp})
is
\begin{align}
k_A+k_B+k_C \geq 2R+2
\end{align}
Several alternate proofs of this fundamental result have been given \cite{tBS02,JS04,SS07,Rho10,Lan}.
Sidiropoulos and Bro \cite{SB00} extended this result to $\ell$-order tensors.
Let $T$ be a $\ell$-order tensor with decomposition
\begin{align*}
T = \sum_{r=1}^R \bigotimes_{j=1}^\ell U_r^{(j)}
\end{align*}
Then the decomposition is unique if
\begin{align}\label{eq:intro:kr}
\sum_{j=1}^\ell k_{U^{(j)}} \geq 2R + (\ell-1)
\end{align}

We give a robust version of of Kruskal's uniqueness theorem for decomposition of 3-tensors.
To this end, we need a natural robust analogue of Kruskal rank: we say that $\krk{\ka}(A)\ge k$ 
if every submatrix of $A$ formed by $k$ of its columns has minimum singular value at least $1/\ka$.
A matrix is called bounded if its column vectors have bounded length. Finally, we measure closeness between two tensors or two matrices by the Frobenius norm of their difference. Please see Section~\ref{sec:prelim} for precise definitions.

Our first result shows that any tensor with bounded decomposition that satisfies the robust Kruskal condition has a unique decomposition upto small error  (formal statement in Section~\ref{sec:prelim}):

\vspace{4pt}

\noindent\textbf{Informal Theorem. }{\it If any order $3$ tensor $T$ has a bounded rank $R$ decomposition $\tens{A}{B}{C}$, where the robust $\krk{\ka}$ $k_A,k_B,k_C$ satisfy $k_A+k_B+k_C \ge 2R+2$, then any decomposition $\tens{A'}{B'}{C'}$ that is $\eps$-close to $T$ has $A',B',C'$ being individually $\eps'$-close to $A, B$ and $C$ respectively when $\eps<\eps' \cdot \poly(R,n,\ka)$. 
}\\

A similar theorem (see Theorem~\ref{thm:unique:gen}) also holds for higher order tensors and the analogous robust Kruskal rank condition is exactly \eqref{eq:intro:kr} where $k_{\spc{U}{j}}$ corresponds to the robust Kruskal rank of $\spc{U}{j}$. Note that when all the $\spc{U}{j}$ have the same rank, the robust Kruskal condition becomes weaker for higher order tensors.

\vspace{2pt}
\noindent{\bf Why is it non-trivial to obtain a robust version from existing proofs?}
Kruskal's theorem gives conditions under which the {\em components} of
a tensor decomposition can be identified uniquely. However the proofs that we
are aware of strongly use inductive lemmas which prove that subsets of
the components of one decomposition have to necessarily belong in any
other potential decomposition, and use them to conclude that any two
decompositions are in fact the same. When working with representations
that are only nearly equal, these inductive arguments typically
accumulate errors in each step, thereby requiring the initial error to
be exponentially small in order to reach the desired conclusion.
Such a result would not be of any value for establishing polynomial sample complexity bounds,
since the sample size would need to be exponentially large for the empirical moment tensors
to approximate the true moment tensor within such a low error.
We overcome this issue by using arguments that are purely
combinatorial whenever possible, and carefully avoiding a loss at each
step.

Since finding low-rank decomposition of tensors is of great practical interest, it is natural to study algorithms for this problem. While this and many related problems are NP-hard in general \cite{HL13}, we give an algorithm which given an approximation to a tensor, finds an approximate low-rank decomposition in time exponential only in the rank (and not the dimensions of the tensor). 

\vspace{4pt}
\noindent\textbf{Informal Theorem. }{\it Given a tensor with a bounded, rank $R$ decomposition up to an error $\eps$, we can find a rank $R$ approximation with error $O(\eps)$ in time $\exp(R^2 \log(n/\eps))\poly(n)$.}

This can be viewed as a tensor analog of low-rank approximation, which is very well-studied for matrices. Note that our algorithm does not require the promised decomposition to have additional {\em well-conditioned} properties. If we additionally have such guarantees (for e.g., that the sum of \kr of the components is high), then Theorem~\ref{thm:unique:gen} implies that the algorithm finds this particular decomposition (up to a small error).

\subsection{Latent Variable Models}

We now describe some of the latent variable models that our results are applicable to.
We will formally state the identifiability and algorithmic results we obtain for each of these in Section~\ref{sec:applications}.

Consider a simple mixture-model, where each sample is generated from mixture of $R$ distributions $\{\calD_r\}_{r \in [R]}$, with mixing probabilities $\{w_r\}_{r \in [R]}$. Here the latent variable $h$ corresponds to the choice of distribution and it can have $[R]$ possibilities. First the distribution $h=r$ is picked with probability $w_r$, and then the data is sampled according to $\calD_r$, which has mean $\mu_r \in \R^n$. Let $M_{n \times R}$ represent the matrix of these $R$ means. This setting captures many latent variable models including topic models, Hidden Markov Models (HMMs), gaussian mixtures etc. 

\subsubsection*{Multi-view Mixture Model}

Multi-view models are mixture models with a discrete latent variable $h \in [R]$, such that $\Prb{h= r}= w_r$.
We are given multiple observations or views $\spc{x}{1}, \spc{x}{2}, \dots,\spc{x}{\ell}$ that are conditionally independent given the latent variable $h$, with $\E{\spc{x}{j} \vert h=r}= \spc{\mu}{j}_r$.
Let $\spc{M}{j}$ be the $n \times R$ matrix whose columns are the means $\{\spc{\mu}{j}_r\}_{r \in [R]}$.
The goal is to learn the matrices $\{\spc{M}{j}\}_{j \in [\ell]}$ and the mixing weights $\{w_r\}_{r \in [R]}$.

Multi-view models are very expressive, and capture many well-studied models like Topic Models \cite{AHK12}, Hidden Markov Models (HMMs) \cite{MR06,AMR09,AHK12}, random graph mixtures \cite{AMR09}, and the techniques developed for this class have also been applied to phylogenetic tree models \cite{Cha96,MR06} and certain tree mixtures \cite{AHHK12}. 

\subsubsection*{Exchangeable (single) Topic Model}

The simplest latent variable model that fits the multi-view setting is the Exchangeable Single Topic model as given in \cite{AHK12}. This is a simple bag-of-words model for documents, in which the words in a document are assumed to be exchangeable. This model can be viewed as first picking the topic $r \in [R]$ of the document, with probability $w_r$.  Given a topic $r \in [R]$, each word in the document is sampled independently at random according to the probability distribution $\mu_r \in \R^n$ ($n$ is the dictionary size). In other words, the topic $r \in [R]$ is a latent variable such that the $\ell$ words in a document are conditionally i.i.d given $r$. 

The views in this case correspond to the words in a document. This is a special case of the multi-view model since the distribution of each of the views $j \in [\ell]$ is identical. 

\subsubsection*{Hidden Markov Models}

Hidden Markov Models (HMMs) are extensively used in speech recognition, image classification, bioinformatics etc\cite{Edd96,GY08}. We follow the same setting as in \cite{AMR09}: there is a hidden state sequence $Z_1,Z_2,\dots,Z_m$ taking values in $[R]$, that forms a stationary Markov chain $Z_1 \rightarrow Z_2 \rightarrow \dots \rightarrow Z_m$ with transition matrix $P$ and initial distribution $w=\{w_r\}_{r \in [R]}$ (assumed to be the stationary distribution). The observation $X_t$ 
is represented by a vector in $\spc{x}{t} \in \R^n$. Given the state $Z_t$ at time $t$, $X_t$ (and hence $\spc{x}{t}$) is conditionally independent of all other observations and states. The matrix $M$ (of size $n \times R$) represents the probability distribution for the observations: the $r^{th}$ column $M_r$ represents the probability distribution conditioned on the state $Z_t=r$ i.e.
$$\forall r \in [R], t \in [m], i \in [n], \quad \Prb{X_t = i \vert Z_t = r}= M_{ir}.$$

As mentioned previously, in many important applications of HMMs, $n$ is much smaller than $R$.
e.g. in image classification, the commonly used SIFT features \cite{Low99} are 128 dimensional, while the
number of image classes is much larger, e.g. 256 classes in the Caltech-256 dataset \cite{GHP07} and
several thousands in the case of ImageNet \cite{image-net}.
Similarly, in speech recognition, the features of an audio signal are typically based on
mel-frequency cepstral coefficients (MFCCs) or an encoding called perceptual linear prediction (PLP)
that incorporates psychoacoustic constraints \cite{GY08}, e.g. these are used to obtain a
39 dimensional feature vector in the popular HTK toolkit for building HMMs for speech recognition \cite{HTK-book, HTK}.
On the other hand, the number of states in these HMMs is much larger.
Further, in some other applications, even when the feature vectors lie in a large dimensional space ($n\gg R$), the set of relevant features or the effective feature space could be a space of much smaller dimension ($k < R$), that is unknown to us. 

\subsubsection*{Mixtures of Spherical Gaussians}

Learning mixtures of Gaussians has a long and rich history -- our overview is necessarily brief and
focuses on work relevant to our results.
We consider the setting where we have a mixture of $R$ spherical gaussians in $\R^n$, with mixing weights $w_1, w_2, \dots w_R$, means $\mu_1, \mu_2, \dots, \mu_r$, and the common variance $\sigma^2$.
Much work on this problem needs certain separation guarantees between the centers
\cite{Das99, AK01, VW04, AM05, DS07, KK10, AS12}. 
Recently, moment methods were developed for arbitrary gaussians \cite{KMV10, MV10, BS10},
albeit with sample complexity and running time exponential in $R$ -- such dependence is necessary in 
general.
Recent work \cite{AGHKT12,HK13} developed efficient algorithms for special cases of this problem 
without needing any separation assumptions. 
These methods, based on tensor decompositions, need the condition that the means are linearly independent and hence necessarily $n \geq R$.
(additionally, the matrix of means need to be well conditioned, i.e. the means should not be close
to a low dimensional subspace).


We apply our results on tensor decompositions to many of the latent variable models described above.
Here is a representative result for multi-view models that applies when the dimension of the
observations ($n$) is $\delta R$ where $\delta$ is a small positive constant and $R$ is the size of 
the range of the hidden variable, and hence the rank of the associated tensors.
In order to establish this, we apply our robust uniqueness result to the 
$\ell^{th}$ moment tensor for $\ell=\lceil 2/\delta \rceil+1$.

\vspace{4pt}

\noindent\textbf{Informal Theorem. }{\it For a multi-view model with $R$ topics or distributions, such that each of the parameter matrices $\spc{M}{j}$ has robust \kr of at least $\delta R$ for some constant $\delta$, we can learn these parameters upto error $\eps$ with high probability using $\poly_\delta(n,R)$ samples. Further, these parameters can be approximately computed in time $\exp_{\delta}\left(R^2 \log(n/\eps)\right) \poly(n)$ time. }



Polynomial identifiability was not known previously for these models in the settings that we
consider.
Moreover, except for the well studied setting of mixtures of Gaussians, no provably good algorithms
were known (even with running time $\exp(\poly(R))$).

For mixtures of Gaussians, our results shed more light on polynomial identifiability: 
the algorithm of \cite{AGHKT12,HK13} shows how to identify mixtures of
(spherical) Gaussians efficiently when we have $R$ Gaussians in $d$
dimensions, when the means satisfy certain well-conditioned properties
(which in particular requires $d \ge R$). When $d=2$, Moitra and
Valiant \cite{MV10} rule out polynomial identifiability by giving two distributions
for which we require exponentially many samples to distinguish one from the other. Thus it
is natural to ask what happens in between, when $d<R$, but is not too small.
Our results imply that a mixture of $R$ Gaussians of
{\em known variance} in a $\del R$ dimensional space (any $\del>0$) can be
identified with polynomially many samples.

\subsection{Overview of Techniques}


\paragraph{Robust Uniqueness of Tensor decompositions.} 
The main technical contribution of our paper is the Robust Uniqueness theorem for Tensor decompositions. Our proof broadly follows the outline of Kruskal's original proof \cite{Kru77}: It proceeds by establishing a certain \emph{Permutation lemma}, which gives necessary conditions to conclude that the columns of two matrices are permutations of each other (up to scaling). Given two decompositions $\tens{A}{B}{C}$ and $\tens{A'}{B'}{C'}$ for the same tensor, it is shown that $A, A'$ satisfy the conditions of the lemma, and thus are permutations of each other. Finally, it is shown that the three permutations for $A,B$ and $C$ (respectively) are identical. 
To prove the robust uniqueness theorem, the key ingredient is a robust version of the permutation lemma. 

The first step in our argument is to prove that if $A$, $B$, $C$ are ``well-conditioned'' (i.e., satisfy the \kr conditions of the theorem), then any other ``bounded'' decomposition which is $\eps$-close is also well-conditioned. This step is crucial to our argument, while an analogous step was not explicitly needed for the proofs of exact uniqueness theorem.\footnote{Note that the uniqueness theorem, in hindsight, establishes that the other decomposition is also well-conditioned.} Besides, this statement is interesting in its own right: it implies, for instance, that there cannot be a smaller rank  (bounded) decomposition. 

The second and most technical step is to prove the robust permutation lemma.
The (robust) Permutation lemma needs to establish that for every column of $C'$, there is some column of $C$ close to it. Kruskal's proof \cite{Kru77} roughly uses downward induction to establish the following claim: for every set of $i \le$ \kr columns of $C'$, there are at least as many columns of $C$ that are in the span of the chosen vectors. The downward induction infers this by considering intersections of columns close to $i+1$ dimensional spaces. 

The natural analogue of this approach would be to consider columns of $C$ which are $\eps$-close to the spans of subsets of columns of $C'$. However, the inductive step involves considering combinations and intersections of the different spans that arise, and such arguments do not seem very tolerant to noise. In particular, we lose a factor of $\ka n$ in each iteration, i.e., if the statement was true for $i+1$ with error $\eps_{i+1}$, it will be true for $i$ with error $\eps_i = \ka n \cdot \eps_{i+1}$. Since $k$ steps of downward induction need to be unrolled, we recover a robust permutation lemma only when the error $<1/(\ka n)^k$ to start with, which is exponentially small since $k$ is typically $\Theta(n)$. 

We overcome this issue by showing a different, more tricky inductive statement, whereby we do not lose any error in the recursion. This is described in Section~\ref{sec:perm-lemma}. To carry forth this argument we crucially rely on the fact that $C'$ is also ``well-conditioned'' and other observations. 

\paragraph{Algorithms for low-rank tensor decompositions.}

At a high level, our algorithm for finding a rank $R$ approximation proceeds by finding a small ($O(R)$) dimensional space and then exhaustively searching, which takes time $\exp(R^2 \log n)\poly(n)$. Note that a naive exhaustive search using an $\eps$-net in the entire $n$ dimensional space would incur a run time of $\exp(R n)\poly(n)$, which is much worse if $n \gg R$.

Suppose the best rank $R$ approximation to an input tensor has error $\eps$. We first find an $R$-dimensional space for each of the (three) dimensions, so that there is an $O(\eps)$-close rank $R$ decomposition that comprises vectors only from the corresponding $R$-dimensional spaces.  We note that the spaces we find need not correspond to the span of the components in the optimum decomposition, but they suffice to obtain an $O(\eps)$ approximation. Another feature of the algorithm is that it does not assume that the tensor has an approximate ``well conditioned'' decomposition, and assumes only boundedness.

\subsection{Related Work}

While our applications to learning latent variable models are inspired by the works of \cite{AHK12, AGHKT12}, our results are significantly different, particularly from a tensor decomposition perspective. Anandkumar et al \cite{AGHKT12} give algorithms for tensors which have a \emph{symmetric orthogonal decomposition}, i.e. a decomposition of the form $\sum_{r=1}^R A_r \otimes A_r \otimes A_r$ where
the vectors $A_r$ are orthogonal. 
In general, a rank-$R$ tensor may not have any orthogonal decomposition. 
Note that any tensor in $n \times n \times n$ dimensions, which has rank $R>n$ can not have an orthogonal decomposition. While this is one source of intractability for general tensor decompositions \cite{HK12}, we crucially use such tensors of rank $R>n$ to give polynomial identifiability beyond the non-degenerate range ($R \le n$).

For various latent variable models, in the non-degenerate setting (where the number of mixtures/ topics $R$ is larger than the dimension of the space $n$), Anandkumar et al \cite{AGHKT12} use order $3$ tensors given by the third moment tensor to identify the hidden parameters. In these tensors, each rank-$1$ component corresponds to a hidden parameter, like one of the means. While these parameters may not be orthogonal, a certain ``\emph{whitening}'' transform of the space \cite{AHK12,HK12} produces a new instance in which these means are now orthogonal. 
For this they crucially rely on two assumptions:
\begin{itemize}
\item The $n \times R$ matrix of the means has rank $\ge R$ (and well conditioned). This of course needs $R \le n$. 
\item The algorithm has access to the second moment tensor\footnote{This is certainly a valid assumption when learning latent variable models}. This assumption will not hold in the case of the general problem of tensor decompositions. 
\end{itemize}

Finally, in the context of learning latent variable models, we go beyond the non-degeneracy barrier and get polynomial identifiability even when $n=\delta R < R$. One interesting aspect of our results is that we use successively higher $O(1)$-moments to handle larger values of $R$ (hidden topics/ mixtures). This smooth tradeoff\footnote{Note that the $R^{th}$ moment is sufficient to identify the parameters typically \cite{BS10,MV10,FOS06}.} is in contrast to the works of \cite{AHK12,HK12,AGHKT12}, where they seem to get no additional advantage out of higher moments (larger than $3$). 
Further, even when using third moments, \cite{AHK12,HK12,AGHKT12} only obtain polynomial identifiability when $R \le n$, whereas we obtain polynomial identifiability till $R= 3n/2-1$. 
On the other hand, since we argue about identifiability directly through uniqueness theorems for tensors, it allows us to handle larger values of $R$. 

We also mention work on PAC learning of mixtures of $k$ product distributions (see e.g. \cite{FOS05,FOS06}) that typically run in $\exp(k)\poly(n)$ time and produce a distribution that is statistically
close to the underlying distribution -- however they do not recover the actual mixture components themselves.

%% file: prelim.tex
We start with basic notation on tensors which we will use throughout the paper. We then state our results formally in these terms, and place them in context. In the process, we will see some intriguing properties of tensors (relevant to our results) which distinguish them from matrices.

\subsection{Notation and Preliminaries}

Tensors are higher dimensional arrays. An $\ell$th order, or $\ell$-dimensional tensor is an element in $\R^{n_1 \times n_2 \times \dots \times n_\ell}$, for positive integers $n_i$. Tensors have classically been defined over complex numbers for certain applications, but we will consider only real tensors.

A concept that plays a crucial role for us is that of the {\em rank} of a tensor. For this, we first define a rank-$1$ tensor as a product $a^{(1)} \opr a^{(2)} \opr \dots \opr a^{(\ell)}$, where $a^{(i)}$ is an $n_i$ dimensional vector. We can now define the rank.

\begin{definition}[Tensor rank, Rank $R$ decomposition]
The {\em rank} of a tensor $T \in \R^{n_1 \times n_2 \times \dots \times n_\ell}$ is defined to be the smallest $R$ for which there exist $R$ rank-1 tensors $T^{(i)}$ whose sum is $T$.

A rank-$R$ decomposition of $T$ is given by a set of matrices $U^{(1)}, U^{(2)}, \dots, U^{(\ell)}$ with $U^{(i)}$ of dimension $n_i \times R$, such that we can write $T = [U^{(1)} ~ U^{(2)} ~ \dots ~U^{(\ell)}]$, which is defined by
\[ [U^{(1)} ~ U^{(2)} ~ \dots ~U^{(\ell)}] := \sum_{r=1}^R U^{(1)}_r \opr U^{(2)}_r \opr \dots \opr U^{(\ell)}_r, \]
where we use the notation $A_r$ to denote the $r$th column vector of matrix $A$.
\end{definition}


Third order tensors (or $3$-tensors) play a central role in understanding properties of tensors in general (as in many other areas of mathematics, the jump in complexity occurs most dramatically when we go from two to three dimensions, in this case from matrices to $3$-tensors). For $3$-tensors, we will often write the decomposition as $\tens{A}{B}{C}$, where $A, B, C$ have dimensions $\na, \nb, \nc$ respectively.

\begin{definition}[$\eps$-close]
Two tensors, represented by $T_1 = [U^{(1)} ~ U^{(2)} ~ \dots ~U^{(\ell)}]$ and $T_2 = [V^{(1)} ~ V^{(2)} ~ \dots ~V^{(\ell)}]$ (of potentially different rank) are said to be $\eps$-close if the Frobenius norm of the difference is small, i.e.,
\[ \norm{[U^{(1)} ~ U^{(2)} ~ \dots ~U^{(\ell)}] - [V^{(1)} ~ V^{(2)} ~ \dots ~V^{(\ell)}]}_F \le \eps\]
We will sometimes write this as $T_1 =_{\eps} T_2$.
\end{definition}

Unless mentioned specifically, the errors in the paper will be $\ell_2$ (or Frobenius norm, which is the square root of the sum of squares of entries in a matrix/tensor), since they add up conveniently.

\begin{definition}[$\ub$-boundedness]
An $n \times R$ matrix $A$ is said to be $\ub$-bounded if each of the columns has length at most $\ub$, for some parameter $\ub$.

A tensor represented as above, $[U^{(1)} ~ U^{(2)} ~ \dots ~U^{(\ell)}]$, is $(\ub_1,\ub_2,\dots, \ub_\ell)$-bounded if the matrix $U^{(i)}$ is $\ub_i$ bounded for all $i$.
\end{definition}

We next define the notion of Kruskal rank, and its robust counterpart.
\begin{definition}[Kruskal rank, $\krk{\ka}(.)$]
Let $A$ be an $n \times R$ matrix. The \kr (or Kruskal rank) of $A$ is the largest $k$ for which {\em every} set of $k$ columns of $A$ are linearly independent.

Let $\ka$ be a parameter. The $\ka$-{\em robust} k-rank is denoted by $\krk{\ka}(A)$, and is the largest $k$ for which every $n \times k$ sub-matrix $A_{|S}$ of $A$ has $\sigma_k (A_{|S}) \ge 1/\ka$.
\end{definition}

Note that we only have a lower bound on the ($k$th) smallest singular value of $A$, and not for example the condition number $\sigma_{\max}/\sigma_k$. This is because we will usually deal with matrices that are also $\rho$-bounded, so such a bound will automatically hold, but our definition makes the notation a little cleaner. We also note that this is somewhat in the spirit of (but much weaker than) the Restricted Isometry Property (RIP)~\cite{CTao05} from the Compressed Sensing literature.

Another simple linear algebra definition we use is the following
\begin{definition}[$\eps$-close to a space]
Let $V$ be a subspace of $\R^n$, and let $\Pi$ be the projection matrix onto $V$. Let $u \in \R^n$. We say that $u$ is $\eps$-close to $V$ if $\norm{u - \Pi u} \le \eps$.
\end{definition}

\paragraph{Other notation.} For $z \in \bbR^{d}$, $\diag(z)$ is the $d \times d$ diagonal matrix with the entries of $z$ occupying the diagonal. For a vector $z \in \R^{d}$, $nz(z)$ denotes the number of non-zero entries in $z$. Further, $nz_{\eps}(z)$ denotes the number of entries of magnitude $\ge \eps$. As is standard, we denote by $\sigma_i(A)$ the $i$th largest singular value of a matrix $A$. Also, we abuse the notation of $\opr$ at times, with $u \opr v$ sometimes referring to a matrix of dimension $dim(u) \times dim(v)$, and sometimes a $dim(u) \cdot dim(v)$ vector. This will always be clear from context.

\paragraph{Normalization.} To avoid complications due to scaling, we will assume that our tensors are scaled such that all the $\ka_A, \ka_B, \dots, $ are $\ge 1$ and $\le \poly(n)$. So also, our upper bounds on lengths $\rho_A, \rho_B, \dots$ are all assumed to be between $1$ and some $\poly(n)$. This helps simplify the statements of our lemmas. 

\paragraph{Error polynomials.} We will, in many places, encounter statements such as ``if $Q_1 \le \eps$, then $Q_2 \le (3n^2 \gamma) \cdot \eps$'', with polynomials $\errp$ (in this case $3n^2 \gamma$) involving the variables $n, R, k_A, k_B, k_C, \ka, \rho, \dots$. In order to keep track of these, we use the notation $\errp_1, \errp_2, \dots$. Sometimes, to refer to a polynomial introduced in Lemma~3.11, for instance, we use $\errp_{3.11}$. Unless specifically mentioned, they will be polynomials in the parameters mentioned above, so we do not mention them each time.


\subsection{Our Results}
We are now ready to formally state the results in our work. The first is a robust version of the uniqueness of decomposition for $3$-tensors.

\begin{theorem}[Unique Decompositions]\label{thm:unique3}
Suppose a rank-$R$ tensor $T=\tens{A}{B}{C}$ is $(\rho_A,\rho_B,\rho_C)$-bounded, with $\krk{\ka_A}(A)=k_A,\krk{\ka_B}(B)=k_B,\krk{\ka_C}(C)=k_C$ satisfying $k_A+k_B+k_C \ge 2R+2$.
Then for every $0<\eps'<1$, there exists
\[ \eps= \eps' / \big( R^6 \errp_{\ref{thm:unique3}}(\ka_A,\rho_A,\rho'_A,n_A)\errp_{\ref{thm:unique3}}(\ka_B,\rho_B,\rho'_B,n_B)\errp_{\ref{thm:unique3}}(\ka_C,\rho_C,\rho'_C,n_C)\big), \]for some polynomial $\errp_{\ref{thm:unique3}}$ such that
for any other $(\rho'_A,\rho'_B,\rho'_C)$-bounded decomposition $\tens{A'}{B'}{C'}$ of rank $R$ that is $\eps$-close to $\tens{A}{B}{C}$, there exists
an ($R \times R$) permutation matrix $\Pi$ and diagonal matrices $\Lambda_A, \Lambda_B, \Lambda_C$ such that 
\begin{equation}\label{eq:uniquegen}
\norm{\Lambda_A \Lambda_B \Lambda_C - I}_F \le \eps' ~\text{and} \quad\nrm{A'-A\Pi \Lambda_A}{F} \le \eps' \quad ( \text{similarly for } B \text{ and } C ) 
\end{equation}
\end{theorem}

We remark that in order to prove the theorem, we did not make any assumptions about the Kruskal ranks of $A', B', C'$. We simply assumed that they are bounded. This is an interesting feature of our proof, and is formalized in Lemma~\ref{lem:conditioned}. Another observation: though we assumed that the decomposition $\tens{A'}{B'}{C'}$ is rank $R$, we really need only an upper bound. This is because we can append zeroes and apply the theorem.

Our next result is a higher dimensional analogue of the above.
\begin{theorem}[Uniqueness of Decompositions for Higher Orders]\label{thm:unique:gen}
Suppose we are given an order $\ell$ tensor (with $\ell \le R$), $T= [\spc{U}{1} ~ \spc{U}{2} ~ \dots ~ \spc{U}{\ell}]$,
where $\forall j \in [\ell]$ the $n_j$-by-$R$ matrix $\spc{U}{j}$ is $\rho_j$-bounded, with $\krk{\ka_j}(\spc{U}{j})=k_j \ge 2$ satisfying 
\[ \sum_{j=1}^{\ell} k_j \ge 2R+\ell-1. \]

Then for every $0<\eps'<1$, there exists $\eps = \left(\spc{\errp}{\ell}_{\ref{thm:unique:gen}}\left (\frac{\eps'}{R}\right) \right) \cdot \left(\prod_{j \in [\ell]} \errp_{\ref{thm:unique:gen}}(\ka_j,\rho_j,\rho'_j,n_j)\right)^{-1}$
such that, for any other $(\rho'_1,\rho'_2,\dots,\rho'_\ell)$-bounded decomposition $[\spc{V}{1}~ \spc{V}{2} ~\dots~ \spc{V}{\ell}]$ which is $\eps$-close to $T$, there exists an $R \times R$ permutation matrix $\Pi$ and diagonal matrices $\{\spc{\Lambda}{j}\}_{j \in [\ell]}$ such that 
\begin{equation}\label{eq:uniquegen:1}
\norm{\prod_{j \in [\ell]}\spc{\Lambda}{j} - I}_F \le \eps' \quad \text{ and} ~\forall j \in [\ell], ~\nrm{\spc{V}{j} - \spc{U}{j} \Pi \spc{\Lambda}{j}}{F} \le \eps'
\end{equation}
Setting $\spc{\errp}{\ell}_{\ref{thm:unique:gen}}(x)=x^{2^{\ell}}$ and 
$\errp_{\ref{thm:unique:gen}}(\ka_j,\rho_j,\rho'_j, n_j)=(\ka_j\rho_j\rho'_jn_j)^{O(1)}$ suffice for the theorem.
\end{theorem}

Since finding a small rank decomposition of a tensor is of great practical interest as we have seen, it is natural to ask if it is possible to compute it efficiently. We can prove:

\begin{theorem}\label{thm:algorithm}
Suppose $T$ is a $3$-tensor which has an (unknown) $\rho$-bounded representation $\tens{A}{B}{C}$, where $A, B, C$ have dimensions $n_A \times R, n_B \times R$ and $n_C \times R$ respectively, for some parameter $\rho$. Then, given a tensor $T'$ which is $\eps$-close to $T$, we can find a rank-$R$ tensor $T''$ (along with its decomposition) which is $5\eps$ close to $T$ in time $\poly(n_A, n_B, n_C) \cdot \exp(R^2 \log(R\rho /\eps))$.
\end{theorem}

We can view the above as an approximation algorithm for the low-rank approximation problem for tensors. We will expound on this viewpoint in Section~\ref{sec:algo}. We also note that although our algorithm is quite simple, it has a running time better than simply trying to guess the $3R$ vectors in the decomposition. The latter typically takes time $\exp(R(n_A + n_B + n_C))$, which could be much worse than our bound for small values of $R$ (which is when the low rank approximation problem is typically interesting).

As we mentioned before, the algorithm does not need the promised decomposition $\tens{A}{B}{C}$ to have large \kr. However, if we are guaranteed that it has additional {\em well-conditioned} properties (for e.g., the sum of \kr of $A, B, C$ is $\ge 2R+2$), then Theorem~\ref{thm:unique:gen} guarantees that the algorithm finds this particular decomposition (up to a small error).

Also, the algorithm extends naturally to higher dimensional tensors: we state this version in Section~\ref{sec:algo}, Theorem~\ref{thm:algo}.

Finally, we show how the above results on tensor decompositions can be used to learn latent variable models with polynomial samples, hence showing polynomial identifiability under some weak conditions involving the \kr{} of the matrices.  We first show polynomial identifiability for the Multi-view mixture model, which captures various latent variable models that are used commonly. 

%

\begin{theorem}[Polynomial Identifiability of Multi-view mixture model]\label{thm:MM}
The following statement holds for any constant integer $\ell$. 
Suppose we are given samples from a multi-view mixture model (see Def~\ref{def:MM}), with the parameters satisfying:
\begin{enumerate}[(a)]
\item For each mixture $r \in [R]$, the mixture weight $w_r > \gamma$.
\item For each $j \in [\ell]$, $\krk{\ka}(\spc{M}{j})\ge k \ge \frac{2R}{\ell}+1$.
\end{enumerate}
then there is a algorithm that given any $\eta>0$ uses 
$N=\spc{\vartheta_{\ref{thm:MM}}}{\ell} \left(\frac{1}{\eta},R,n,\ka,1/\gamma,\Bd\right)$ samples, \\and finds with high probability $\{\spc{\tM}{j}\}_{j \in [\ell]} $ and $\{\tw_r\}_{r \in [R]}$ (upto renaming of the mixtures $\{1,2,\dots,R\}$) such that
\begin{equation}
\forall j \in [\ell], \quad \nrm{\spc{M}{j}-\spc{\tM}{j}}{F} \le \eta \quad \text{ and }\quad \forall r \in [R], ~ \abs{w_r - \tw_r} < \eta 
\end{equation}
Further, this algorithm runs in time $\exp\left(R^2 \ell^2 \left[2^{2\ell} \log(\frac{R\ell}{\eta})+\ell \log(n\cdot \frac{\ka \Bd}{\gamma}) \right] \right)\poly(n)$ time.
\end{theorem}

Polynomial identifiability of the Multi-view mixture model also leads to polynomial identifiability of other latent variable models like topic models and HMMs. The following corollary shows that Hidden Markov models can be learned from polynomial many samples by observing constant number of consecutive time steps under mild conditions involving the \kr (the constant depends on the exact \kr condition). Please refer to section~\ref{sec:applications} to see the implications for other latent variable and mixture models like topic models, mixtures of gaussians etc. 

\vspace{5pt}
\noindent{\bf Corollary~\ref{corr:hmm}} (Polynomial Identifiability of Hidden Markov models).
\emph{ 
The following statement holds for any constant $\delta>0$. 
Suppose we are given a Hidden Markov model with parameters as follows :
\begin{enumerate}[(a)]
\item The stationary distribution $\{w_r\}_{r \in [R]}$ has $\forall r \in [R]~ w_r > \gamma_1$,
\item The observation matrix $M$ has $\krk{\ka}(M)\ge k \ge \delta R$,
\item The transition matrix $P$ has minimum singular value $\sigma_{R}(P) \ge \gamma_2$,
\end{enumerate}
then there is a algorithm that given any $\eta>0$ uses $N=\spc{\vartheta_{\ref{thm:MM}}}{\frac{1}{\delta}+1}\left(\frac{1}{\eta},R,n,\ka,\frac{1}{\gamma_1 \gamma_2} \right)$ samples \\
of $m=2\lceil \frac{1}{\delta} \rceil +3$ consecutive observations (of the Markov Chain), 
and finds with high probability, $P',M'$ and $\{\tw_r\}_{r \in [R]}$ such that
\begin{equation}
\nrm{M-M'}{F} \le \eta, \quad \nrm{P-P'}{F} \le \eta ~ \text{ and } \forall r \in [R], ~ \abs{w_r - \tw_r} < \eta 
\end{equation}
Further, this algorithm runs in time $n^{O_{\delta}(R^2 \log (\frac{1}{\eta \gamma_1}))}\left(n \cdot \frac{\ka}{\gamma_1 \gamma_2}\right)^{O_\delta(1)}$ time.
}\\

\vspace{3pt}

Note that the above results shows polynomial identifiability (for constant $\delta>0$), and additionally gives an algorithm which takes time $n^{O_{\delta}(R^2)}\poly(n,\ka,R)$ for inverse polynomial error. To the best of our knowledge such algorithmic results with only a polynomial dependence on $n$ were not known for learning HMMs and topic models. 


\subsection{Auxiliary lemmas}
In our proofs we will require several simple (mostly elementary linear algebra) lemmas. The Section~\ref{sec:app:linear-algebra} is a medley of such lemmas. Most of the proofs are reasonably straightforward, and thus we place them in the Appendix.

%% file: robustkruskal.tex
\newcommand{\tka}{\tilde{\ka}}



First we consider third order tensors and prove Theorem~\ref{thm:unique3} (Sections~\ref{sec:unique3:1} and \ref{sec:unique3:2}). Our proof broadly follows along the lines of Kruskal's original proof of the uniqueness theorem~\cite{Kru77}. The key ingredient, which is a robust version of the so-called {\em permutation lemma} is presented in Section~\ref{sec:perm-lemma}, since it seems interesting its own right. Finally we will see how to reduce the case of higher order tensors, i.e. Theorem~\ref{thm:unique:gen}, to that of third order tensors (Section~\ref{sec:unique-higher}).

\subsection{Uniqueness Theorem for Third Order Tensors}\label{sec:unique3:1}
The proof of Theorem~\ref{thm:unique3} broadly has two parts. First, we prove that if $[A, B, C] = [A', B', C']$, then $A$ is a permutation of $A'$, $B$ of $B'$, and $C$ of $C'$. Second, we prove that the permutations in the (three) different ``modes'' (or dimensions) are indeed equal. 
Let us begin by describing a lemma which is key to the first step.

\paragraph{The Permutation Lemma}
This is the core of Kruskal's argument for the uniqueness of tensor decompositions. Given two matrices $X$ and $Y$, how does one conclude that the columns are permutations of each other? Kruskal gives a very clever sufficient condition, involving looking at {\em test vectors} $w$, and considering the number of non-zero entries of $w^T X$ and $w^T Y$. The intuition is that if $X$ and $Y$ are indeed permutations, these numbers are precisely equal for all $w$.

Kruskal proves that if this sufficient condition holds, then $X$ and $Y$ must have columns which are permutations of each other, up to scaling. More precisely, suppose $X, Y$ are $n \times R$ matrices of rank $k$. Let $nz(x)$ denote the number of non-zero entries in a vector $x$. The lemma then states that if for all $w$, we have
\[ nz( w^T X) \le R-k+1 ~~\implies ~~ nz(w^T Y) \le nz(w^T X), \]
then the matrices $X$ and $Y$ have columns which are permutations of each other up to a scaling. That is, there exists an $R \times R$ permutation matrix $\Pi$, and a diagonal matrix $\Lambda$ s.t. $Y = X \Pi \Lambda$.

We prove a robust version of this lemma, stated as follows (recall the definition of $nz_{\eps}(.)$, Section~\ref{sec:prelim})

\begin{lemma}[Robust permutation lemma]\label{lem:permutation-poly}
Suppose $X, Y$ are $\rho$-bounded $n \times R$ matrices such that $\krk{\ka}(X)$ and $\krk{\ka}(Y)$ are $\ge k$, for some integer $k \ge 2$. Further, suppose that for $\eps < 1/\errp_{\ref{lem:permutation-poly}}$, the matrices
satisfy:
\begin{equation}\label{eq:perm-cond} \forall w \text{ s.t. }\quad  nz(w^T X) \le R-k+1 \text{, we have }  nz_{\eps}(w^T Y) \le nz(w^T X),
\end{equation}
then there exists an $R \times R$ permutation matrix $\Pi$, and a diagonal matrix $\Lambda$ s.t. $X$ and $Y$ satisfy $\norm{ X - Y \Pi \Lambda }_F < \errp_{\ref{lem:permutation-poly}} \cdot \eps$. In fact, we can pick $\errp_{\ref{lem:permutation-poly}} := (nR^2) \errp_{\ref{lem:key-intersection}}$.
\end{lemma}

\paragraph{Outline of the section.} In the remainder of this section, we will prove that $A'$ is a permutation of $A$, $B'$ of $B$ and $C'$ of $C$. We do this by assuming Lemma~\ref{lem:permutation-poly} for now (it will be proved in Section~\ref{sec:perm-lemma}) and proving that if $\tens{A}{B}{C} =_{\eps} \tens{A'}{B'}{C'}$, then the conditions of the lemma hold for $C', C$ as $X, Y$ in the statement respectively. We can repeat this argument with $A, B$ to obtain the conclusion.


We now state the key technical lemma which allows us to verify that the hypotheses of Lemma~\ref{lem:permutation-poly} hold. It says for any $k_C -1$ vectors of $C'$ there are at least as many columns of $C$ which are close to the span of the chosen columns from $C'$. 

\newcommand{\tensA}{\tens{A}{B}{C}}
\newcommand{\tensAP}{\tens{A'}{B'}{C'}}
\begin{lemma}\label{lem:perm:conditions}
Suppose $A, B, C, A', B', C'$ satisfy the conditions of Theorem~\ref{thm:unique3}, and suppose $\tens{A}{B}{C} =_{\eps} \tens{A'}{B'}{C'}$. Then for any unit vector $x$, we have
\[\forall \eps', \quad \nz_{\eps'}(x^T C') \le R-k_C+1 \implies \nz_{\eps''}(x^T C) \le \nz_{\eps'}(x^T C')\] 
for $\eps'' = \errp_{\ref{lem:perm:conditions}} \cdot (\eps + \eps')$, where $\errp_{\ref{lem:perm:conditions}} := 4R^3 (\ka_A \ka_B \ka_C)^2 \rho_A \rho_B \rho_C (\rho'_A \rho'_B \rho'_C)^2$.
\end{lemma}

\paragraph{Remark.} This lemma, together with its corollary Lemma~\ref{lem:conditioned} will imply the conditions of the permutation lemma. Lemma~\ref{lem:conditioned} lets us conclude that $\krk{\ka \errp}(C') \ge \krk{\ka}(C)$ for some error polynomial $\errp$, which is essential in our proof of the permutation lemma. It also has other implications, as we will see. While the proof of the robust permutation lemma (Lemma~\ref{lem:permutation-poly}) will directly apply this Lemma with $\eps'=0$, we will need the $\eps'>0$ case for establishing Lemma~\ref{lem:conditioned}. 

A key component of the proof is to view the three-dimensional tensor $\tens{A}{B}{C}$ as a bunch of {\em matrix slices}, and argue about the rank (or conditioned-ness) of weighted combinations of these slices. One observation, which follows from the Cauchy-Schwarz inequality, is the following: if $\tens{A}{B}{C}=_{\eps}\tens{A'}{B'}{C'}$, then by taking a combination of slices along the third dimension (with weights given by $x \in \bbR^{n_C}$, i.e., reweighing the $i$th slice by $x_i$ and summing these matrices) we have
\begin{equation} \label{eq:slices}
\forall x \in \bbR^{n_C}, \quad  \norm{A~ \diag(x^T C)~ B^T - A'~ \diag(x^T C')~ (B')^{T}}_F^2 \le \eps^2 \norm{x}_2^2. \end{equation}

We now begin the proof of the Lemma.
\begin{proof}[Proof of Lemma~\ref{lem:perm:conditions}]
W.l.o.g., we may assume that $k_A \ge k_B$ (the proof for $k_A < k_B$ will follow along the same lines). For convenience, let us define $\alpha$ to be the vector $x^T C$, and $\beta$ the vector $x^T C'$. Let $t$ be the number of entries of $\beta$ of magnitude $>\eps'$. The assumption of the lemma implies that $t \le R- k_C +1$. Now from \eqref{eq:slices}, we have
\begin{equation}\label{eq:helpful}
M:= \sum_i \alpha_i A_i \opr B_i = \sum_i \beta_i A_i' \opr B_i' + Z,
\end{equation}
where $Z$ is an error matrix satisfying $\norm{Z}_F \le \eps$. Now, since the RHS has at most $t$ terms with $|\beta_i| > \eps'$, we have that $\sigma_{t+1}$ of the LHS is at most $R\rho_A' \rho_B' \eps' + \eps$. Using the value of $t$, we obtain

\begin{equation}\label{eq:conditions:bound}
\sigma_{R-k_C+2}(M) \le \sigma_{t+1} (M) < \eps+ (R \rho_A' \rho_B')\eps' 
\end{equation}

We will now show that if $x^T C$ has too many co-ordinates which are larger than $\eps''$ then we will contradict \eqref{eq:conditions:bound}. One tricky case we need to handle is the following: while each of these non-negligible co-ordinates of $x^T C$ will give rise to a large rank-$1$ term, they can be canceled out by combinations of the rank-$1$ terms corresponding to entries of $x^T C$ which are slightly smaller than $\eps''$. Hence, we will also set a smaller threshold $\delta$ and first handle the case when there are many co-ordinates in $x^T C$ which are larger than $\delta$. $\delta$ is chosen so that the terms with $(x^T C)_i< \delta$ can not cancel out any of the large terms ($(x^T C)_i \ge \eps''$).

Define $S_1= \{ i : | (x^T C)_i | > \eps'' \}$ and $S_2= \{ i: |(x^T C)_i | > \delta  \}$, where $\del = \eps''/\errp$ for some error polynomial $\errp=2 R^2 \rho_A \rho_B \rho_C \rho_A' \rho_B' \rho_C' \ka_A \ka_B \ka_C$ (which is always $> 1$). Thus we have $S_1 \subseteq S_2$. We consider two cases.

\noindent\textbf{Case 1: $|S_2| \ge k_B$.}

In this case we will give a lower bound on $\sigma_{R-k_C+2}(M)$, which gives a  contradiction to \eqref{eq:conditions:bound}. The intuition is roughly that $A, B$ have $k_A, k_B$ {\em large} singular values, and thus the product should have enough large ones as well. To formalize this, we use the following well-known fact about singular values of products, which is proved by considering the variational characterization of singular values:
\begin{fact}\label{fact:sylvester}
Let $P, Q$ be matrices of dimensions $p \times m$ and $m \times q$ respectively. Then for all $\ell, i$ such that $\ell \le \min\{p,q\}$, we have
\begin{equation} \label{eq:sylvester}
\sigma_{\ell}(PQ) \ge \sigma_{\ell+m-i}(P) \sigma_{i}(Q)
\end{equation}
\end{fact}

Now, let us view $M$ as $PQ$, where $P = A$, and $Q = \diag(\alpha)B^T$. We will show that $\sigma_{k_B}(Q) \ge \del/ \ka_B$, and that $\sigma_{2R+2-k_B-k_C}(A) \ge 1/\ka_A$. These will then imply a contradiction to~\eqref{eq:conditions:bound} by setting $\ell = R-k_C+2$ and $i = k_B$ since
$$\frac{\delta}{\ka_A \ka_B} = \frac{\eps''}{\errp \ka_A \ka_B} > (R \rho_A' \rho_B' \eps'+\eps) ~\text{ by our choice of }\errp_{\ref{lem:perm:conditions}} .$$ (It is easy to check that $\ell \le \min\{k_A, k_B\} \le \min\{n_A, n_B\}$, and thus we can use the fact above.)

Thus we only need to show the two inequalities above. The latter is easy, because by the hypothesis we have $2R+2-k_B-k_C \le k_A$, and we know that $\sigma_{k_A} (A) \ge 1/\ka_A$, by the definition of $\krk{\ka_A}(A)$. Thus it remains to prove the second inequality. To see this, let $J \subset S_2$ of size $k_B$. Let $B^T_J$ and $Q_J$ be the submatrices of $B^T$ and $Q$ restricted to rows of $J$. Thus we have $Q_J=\diag(\alpha)_J B^T_J$. Because of the Kruskal condition, every $k_B$ sized sub matrix of $B$ is well-conditioned, and thus $\sigma_{k_B}(B_J)=\sigma_{k_B}(B^T_J)\ge 1/\ka_B$.

Further, since $|\alpha_j| >\delta ~ \forall j \in J$, multiplication by the diagonal cannot lower the singular values by much, and we get $\sigma_{k_B} Q_J \ge \del/\ka_B$. This can also be seen formally by noting that $\sigma_{k_B}(\diag(\alpha)_J) \ge \delta$, and applying Fact~\ref{fact:sylvester} with $P=\diag(\alpha)_J, Q=B^T_J$ and $\ell=m=i=k_B$.

Finally, since $Q$ is essentially $Q_J$ along with additional rows, we have $\sigma_{\ka_B}(Q) \ge \sigma_{\ka_B} (Q_J) \ge \delta/\ka_B$.
From the argument earlier, we obtain a contradiction in this case.

\noindent\textbf{Case 2: $|S_2| < k_B$.}

Roughly, by defining $S_1, S_2$, we have divided the coefficients $\alpha_i$ into large ($\ge \eps''$), small, and tiny ($< \del$). In this case, we have that the number of large and small terms together (in $M$, see Eq.~\eqref{eq:helpful}) is at most $k_B$. For contradiction, we can assume the number of large ones is $\ge t+1$, since we are done otherwise. The aim is to now prove that this implies a lower bound on $\sigma_{t+1}(M)$, which gives a contradiction to Eq.~\eqref{eq:conditions:bound}.

Now let us define $M' = \sum_{i \in S_2} \alpha_i (A_i \otimes B_i)$. Thus $M$ and $M'$ are equal up to {\em tiny} terms. Further, let $\Pi$ be the matrix which projects a vector onto the span of $\{ B_i' ~:~ |\beta_i| \ge \eps' \}$, i.e., the span of the columns of $B'$ which correspond to $|\beta_i| \ge \eps'$. Because there are at most $t$ such $\beta_i$, this is a space of dimension $\le t$. Thus we can rewrite Eq.~\eqref{eq:helpful} as
\begin{equation}\label{eq:helpful2}
M' = \sum_{i \in S_1} \alpha_i (A_i \otimes B_i) + \sum_{j \in S_2\setminus S_1} \alpha_j  (A_j \otimes B_j) = \sum_{i=1}^{t} \beta_i (A_i' \opr B_i') + Err,
\end{equation}
where we assumed w.l.o.g. that $|\beta_i| \ge \eps'$ for $i \in [t]$, and $Err$ is an error matrix of Frobenius norm at most $\eps+R (\rho_A \rho_B \del + \rho_A' \rho_B' \eps') \le \eps+(R \rho_A \rho_B \rho_A' \rho_B') (\del + \eps')$.

Now because $|S_1| \ge t+1$, and $\krk{\ka_B} (B) \ge k_B \ge t+1$, there must be one vector among the $B_i$, $i \in S_1$, which has a reasonably large projection orthogonal to the span above, i.e., which satisfies
\[ \norm{ B_i - \Pi B_i}_2 \ge 1/(\ka_B \sqrt{R}). \]
Let us pick a unit vector $y$ along $B_i - \Pi B_i$. Consider the equality~\eqref{eq:helpful2} and multiply by $y$ on both sides. We obtain
\[ \sum_{i \in S_2} \alpha_i \Iprod{B_i}{y} A_i = (Err)y. \]
Thus we have a combination of the $A_i$'s, with at least one coefficient being $> \eps'' / (R\ka_B)$, having a magnitude at most $\norm{(Err)y}_2 < \errp_1 (\del + \eps'+\eps)$, where $\errp_1$ was specified above. \\
Now $k_A \ge k_B \ge |S_2|$. So, we obtain a contradiction by Lemma~\ref{lem:conditioned1} since:
\begin{align*}
\norm{(Err)y}_2 < \errp_1  (\del+\eps'+\eps) &= R \rho_A \rho_B \rho_A' \rho_B' (\del+\eps'+\eps)\\
& = R \rho_A \rho_B \rho_A' \rho_B' (\frac{\eps''}{\errp} +\eps'+\eps)\\
& < \frac{1}{\ka_A} \cdot \frac{\eps''}{ R \ka_B}
\end{align*}
The last inequality follows because $\errp = 2 R^2 \rho_A \rho_B \rho_C \rho_A' \rho_B' \rho_C' \ka_A \ka_B \ka_C$. \\
This completes the proof in this case, hence concluding the proof of the lemma.
\end{proof}

The next lemma uses the above to conclude that $\krk{\errp \ka}(C') \ge \krk{\ka}(C)$, for some polynomial $\errp$.


\begin{lemma}\label{lem:conditioned}
Let $A, B, C, A', B', C'$ be as in the setting of Theorem~\ref{thm:unique3}. Suppose $\tensA =_{\eps} \tensAP$, with 
\[ \eps < 1/\errp_{\ref{lem:conditioned}}, \text{ where } \errp_{\ref{lem:conditioned}} = R \ka_A \ka_B \ka_C \errp_{\ref{lem:perm:conditions}} = 4R^4 \ka_A^3 \ka_B^3 \ka_C^3 \rho_A \rho_B \rho_C (\rho'_A \rho'_B \rho'_C)^2.\]
Then $A',B',C'$ have $\krk{\ka'}$ to be at least $k_A,k_B,k_C$ respectively, where $\ka' := \errp_{\ref{lem:conditioned}}$.
\end{lemma}

\paragraph{Remark.} The lemma implies that if $T$ has a well-conditioned decomposition which satisfies the Kruskal conditions, then any other bounded decomposition which is a  sufficiently good approximation should also be reasonably well-conditioned. Further, it says that the decomposition $\tens{A'}{B'}{C'}$ can not be of rank $< R$. Otherwise, we could add some zero-columns to each of $A',B',C'$ and apply this lemma to conclude \kr of $A'$ is $\ge 2$, a contradiction if there exists a zero column. 



\begin{proof}
By symmetry, let us just show this for matrix $C'$ (dimensions $n \times R$), and let $k=k_C$ for convenience.
We need to show that every $n$-by-$k$ submatrix of $C'$ has minimum singular value $\ge \delta = 1/\ka'_C$.

For contradiction let $C'_S$ be the submatrix corresponding to the columns in $S$ ($|S|=k$), such that 
$\sigma_{k}(C'_S) < \delta$. Let us consider a left singular vector $z$ which corresponds to $\sigma_{k} (C'_S)$, and suppose $z$ is normalized to be unit length. Then we have
\[ \sum_{i \in S} \Iprod{z}{C'_i}^2 < \delta^2 \]
Thus $|\Iprod{z}{C'_i}| < \del$ for all $i \in S$, so we have $\nz_{\delta}(z^T C') \le n-k$. Now from Lemma~\ref{lem:perm:conditions}, we have 
\[ \nz_{\eps_1}(zC) \le n-k, \text{ where } \eps_1 = \errp_{\ref{lem:perm:conditions}} (\eps + \del). \] 

Let $J$ denote the set of indices in $z^T C$ which are $<\eps_1$ in magnitude (by the above, we have $|J| \ge k$). Thus we have $\norm{zC_J}_2 < R \eps_1$, which leads to a contradiction if we have $\krk{1/(R \eps_1)} (C) \ge k$.

Since this is true for our choice of parameters, the claim follows.
\end{proof}

Once we have the lemmas above, let us check that the conditions of the robust permutation lemma hold with $C', C$ taking the roles of $X, Y$ in Lemma~\ref{lem:permutation-poly}, and $k= k_C$, and $\ka = \errp_{\ref{lem:conditioned}} \cdot \ka_C$. From Lemma~\ref{lem:conditioned}, it follows that $\krk{\ka}(C)$ and $\krk{\ka}(C')$ are both $\ge k$, and setting $\eps'=0$ in Lemma~\ref{lem:perm:conditions}, the other condition of Lemma~\ref{lem:permutation-poly} holds. Thus we can conclude that there exists a permutation matrix $\Pi_C$ and a diagonal matrix of scalars $\Lambda_C$ such that $\norm{C' - C \Pi_C \Lambda_C}_F$ is {\em small}. We will see the quantitative details in what follows.

\subsection{Wrapping up the proof}\label{sec:unique3:2}
We are now ready to complete the robust Kruskal's theorem. From what we saw above, the main part that remains is to prove that the permutations in the various {\em dimensions} are equal.

\begin{proof}[Proof of Theorem~\ref{thm:unique3}]
Suppose we are given an $\eps' <1$ as in the statement of the theorem. For a moment, suppose $\eps$ is small enough, and $A, B, C, A', B', C'$ satisfying the conditions of the theorem produce tensors which are $\eps$-close. 

From the hypothesis, note
that $k_A,k_B,k_C \ge 2$ (since $k_A,k_B,k_C \le R$, and $k_A + k_B + k_C \ge 2R+2$). Thus from the Lemmas~\ref{lem:conditioned} and \ref{lem:perm:conditions} (setting $\eps'=0$), we obtain that $C, C'$ satisfy the hypothesis of the Robust permutation lemma (Lemma~\ref{lem:permutation-poly}) with $C', C$ set to $X, Y$ respectively, and the parameters
\[ ``\ka" := \errp_{\ref{lem:conditioned}} ~; ~~ ``\eps" := \errp_{\ref{lem:perm:conditions}} \eps.\]

Hence, we apply Lemma~\ref{lem:permutation-poly} to $A$,$B$ and $C$, and get that there exists permutation matrices $\Pi_A$, $\Pi_B$ and $\Pi_C$ and scalar matrix $\Lambda_A, \Lambda_B, \Lambda_C$ such that for $\eps_2 = \errp_{\ref{lem:permutation-poly}} \errp_{\ref{lem:perm:conditions}} \cdot \eps$,
\begin{equation} \label{eq:unique3:1}
\nrm{ A' - A\Pi_A\Lambda_A}{F} < \eps_2, ~ \nrm{ B' - B\Pi_B\Lambda_B}{F} < \eps_2 ~\text{ and } \nrm{ C' - C\Pi_C\Lambda_C}{F} < \eps_2
\end{equation}

We now need to prove that these three permutations are in fact identical, and that the scalings multiply to the identity (up to small error).

\noindent{\bf To show $\Pi_A=\Pi_B=\Pi_C$:}

Let us assume for contradiction that $\Pi_A \ne \Pi_B$. We will use an index where the permutations disagree to obtain a contradiction to the assumptions on the \kr.

For notational convenience, let $\pi_A:[R] \rightarrow [R]$ correspond to the permutation given by $\Pi_A$, with $\pi_A(r)$ being the column that $A'_r$ maps to. Permutation $\pi_B:[R]\rightarrow [R]$ similarly corresponds to $\Pi_B$.  Using \eqref{eq:unique3:1} for $A$ we have
\begin{align*}
\nrm{ \sum_{r \in [R]} \left(A'_r - \Lambda_A(r) A_{\pi_A(r)}\right) \otimes B'_r \otimes C'_r}{F} &\le \sum_{r \in [R]} \nrm{\left(A'_r - \Lambda_A(r) A_{\pi_A(r)}\right) \otimes B'_r \otimes C'_r}{F}\\
&\le \eps_2 \sqrt{R} \rho'_B \rho'_C ~\text{ using Cauchy-Schwarz}
\end{align*}

By a similar argument, and using triangle inequality ( along with $\eps_2 \le 1 \le \rho'_B$) we get
$$\nrm{ \sum_{r \in [R]} A'_r \otimes B'_r \otimes C'_r - \sum_{r \in [R]} \Lambda_A(r) \Lambda_B \cdot A_{\pi_A(r)}\otimes B_{\pi_B(r)} \otimes C'_r }{F} \le 2 \eps_2 \sqrt{R} (\rho'_B  \rho'_C + \rho'_A \rho'_C)$$

Let us take linear combinations given by unit vectors $v$ and $w$, of the given tensor $T=\tens{A}{B}{C}$ along the first and second dimensions. By combining the above inequality along with the fact that the two decompositions are $\eps$-close i.e. $\nrm{ \sum_{r \in [R]} A_r \otimes B_r \otimes C_r - A'_r \otimes B'_r \otimes C'_r}{F} \le \eps$,  we have
\begin{align*}\label{eq:twoslices}
\norm{Z-Z'} &\le \eps_3=\eps + 2 \eps_2 R\rho'_C (\rho'_A + \rho'_B) ~\text{ where } \\
Z=\sum_{r \in [R]} \Iprod{v}{A_r}\Iprod{w}{B_r} C_r &\quad \text{and } Z'=\sum_{r \in [R]} \Lambda_A(r)\Lambda_B(r)\Iprod{v}{A_{\pi_A(r)}}\Iprod{w}{B_{\pi_B(r)}} C'_r
\end{align*}
Note that the $\eps$ term above is negligible compared to the second term involving $\eps_2$. \\
We know that $\pi_A \ne \pi_B$, so there exist $s \ne t \in [R]$ such that $r^*=\pi_A(s)=\pi_B(t)$. 
We will now use this $r^*$ to pick $v$ and $w$ carefully so that the vector $Z'$ is negligible while $Z$ is large. 
We partition $[R]$ into $V,W$ with $|V|=k_A-1$ and $|W|\le k_B-1$, so that $\pi_A(t) \in V$ and $\pi_B(s) \in W$ and for each $r \in [R]-\{s,t\}$, either $\pi_A(r) \in V$ or $\pi_B(r) \in W$. Such a partitioning is possible since $R \le k_A+k_B-2$. 

Let $\calV = \spn(V)$ and $\calW=\spn(W)$. 
We know that $r^*=\pi_A(s) \notin S$ and $r^*=\pi_A(t) \notin T$. \\
Hence, pick $v$ as unit vector along $\Pi^{\perp}_\calV A_{r^*}$ and $w$ as unit vector along $\Pi^{\perp}_\calW B_{r^*}$. 
By this choice, we ensure that $Z'=0$ (since $v \perp \calV$ and $w \perp \calW$).

However, $\krk{\ka_A}(A) \ge k_A$ and $\krk{\ka_B}(B) \ge k_B$, so $\Iprod{v}{A_{r^*}}\Iprod{w}{B_{r^*}} \ge 1/\ka_A \ka_B$ (by Lemma~\ref{lem:conditioned:span}). Further, $|V|=k_A-1$ implies that at most $R-k_A+1 \le k_C-1$ terms of $Z$ is non-zero.
\[ \norm{\sum_{r \in [R]\setminus V} \beta_r C_r} \le \eps_3 ~\text{ where } \beta_r = \Iprod{v}{A_r}\Iprod{w}{B_r} \]

Further, $|\beta_{r^*}| \ge (\ka_A \ka_B)^{-1}$, and since $\krk{\ka_C}(C)=k_C \ge R-|V|+1$, we have a contradiction if $\eps_3 < (\ka_A \ka_B \ka_C)^{-1}$ due to Lemma~\ref{lem:conditioned:span}. This will be true for our choice of parameters. 
Hence $\Pi_A=\Pi_B$, and similarly $\Pi_A=\Pi_C$. Let us denote $\Pi=\Pi_A=\Pi_B=\Pi_C$. In the remainder, we assume $\Pi$ is the identity, since this is without loss of generality.

\noindent{\bf To show $\Lambda_A \Lambda_B \Lambda_C =_{\eps'} I_R$:}

Let us denote $\beta_i=\lambda_A(i)\lambda_B(i) \lambda_C(i)$. From ~\eqref{eq:unique3:1} and triangle inequality, we have as before
$$\nrm{ \sum_{r \in [R]} A'_r \otimes B'_r \otimes C'_r - \sum_{r \in [R]} \Lambda_A(r) \Lambda_B(r) \Lambda_C(r) \cdot A_{\pi_A(r)}\otimes B_{\pi_B(r)} \otimes C_{\pi_C(r)} }{F} \le 5 \eps_2 \sqrt{R} \rho'_A \rho'_B  \rho'_C $$
Combining this with the fact that the decompositions are $\eps$-close we get
\[ \norm{\sum_{r \in [R]} (1-\beta_r) A_r \otimes B_r \otimes C_r} < \eps_4=\eps+5\sqrt{R} \rho'_A \rho'_B \rho'_C\eps_2 \le 6 \sqrt{R} \rho'_A \rho'_B \rho'_C \eps_2. \]
By taking linear combinations given by unit vectors $x,y$ along the first two dimensions (i.e. $xA$ and $yB$) we have 
\[ \norm{ \sum_{r \in [R]} (1-\beta_r)(xA_r)(yB_r) C_r} < \eps_4. \]

We will show each $\beta_r$ is negligible.
Since $R+2 \le k_A+k_B$, let $S, W \subseteq [R]-\{r\}$ be disjoint sets of indices not containing $r$, such that $|S|=k_A-1$ and $|W| \le k_B-1$. Let $\calS=\spn(\{A_j: j\in S\})$ and $\calW=\spn(\{B_j: j\in W\})$.
Let $x$ and $y$ be unit vectors along $\Pi^{\perp}_\calS A_r$ and $\Pi^{\perp}_\calW B_r$ respectively.

Since $\krk{\ka_A}(A) \ge k_A$ and $\krk{\ka_B}(B) \ge k_B$,  we have that $\norm{\Pi^{\perp}_\calS A_r} \ge 1/\ka_A$ (similarly for $B_r$). Hence, from Lemma~\ref{lem:conditioned:span}
\[ (1-\beta_r)(\frac{1}{\ka_A \ka_B})\norm{C_r} < \eps_4 \implies 1-\beta_r < \eps_4 \ka_A \ka_B \ka_C.\]

Thus, $\norm{\Lambda_A \Lambda_B \Lambda_C -I} \le \eps_4 \ka_A \ka_B \ka_C \le \eps'$ (our choice of $\eps$ will ensure this). This implies the theorem. 

Let us now set the $\eps$ for the above to hold (note that $\errp_{\ref{lem:permutation-poly}}$ involves a $\ka$ term which depends on $\errp_{\ref{lem:conditioned}}$)
\[ \eps := \frac{\eps'}{ 6(R\ka_A \ka_B \ka_C)\rho'_A \rho'_B \rho'_C \cdot \errp_{\ref{lem:perm:conditions}} \errp_{\ref{lem:permutation-poly}} },\]
which can easily be seen to be of the form in the statement of the theorem. This completes the proof.
\end{proof}


\input{better-eps2}

\input{robustkruskal-higher}

%% file: better-eps2.tex
\subsection{A Robust Permutation Lemma}\label{sec:perm-lemma}
Let us now prove the robust version of the permutation lemma (Lemma~\ref{lem:permutation-poly}).  Recall that $\krk{\ka}(X)$ and $\krk{\ka}(Y)$ are  $\ge k$, and that the matrices $X, Y$ are $n\times R$.

Kruskal's proof of the permutation lemma proceeds by induction. Roughly, he considers the span of some set of $i$ columns of $X$ (for $i < k$), and proves that there exist at least $i$ columns of $Y$ which lie in this span. The hypothesis of his lemma implies this for $i=k-1$, and the proof proceeds by downward induction. Note that $i=1$ implies for every column of $X$, there is at least one column of $Y$ in its span. Since no two columns of $X$ are {\em parallel}, and the number of columns is equal in $X, Y$, there must be precisely one column, and this completes the proof.

A natural way to mimic this proof is to say: for each set of $i$ columns in $X$, there exist a set of at least $i$ columns in $Y$ which are $\eps_i$ close to the span of the chosen columns in $X$.  The difficulty with this is that we lose a factor of $\ka n$ in each iteration, i.e., if the statement was true for $i+1$ with error $\eps_{i+1}$, it will be true for $i$ with error $\eps_i = \ka n \cdot \eps_{i+1}$. This means that to obtain a small error at the end, we should have started off with error $<1/(\ka n)^k$, which is exponentially small. Thus we need a more tricky inductive statement and additional observations (including Lemma~\ref{lem:conditioned}) to overcome this issue.

We start by introducing some notation. If $V$ is a matrix and $S$ a subset of the columns, we denote by $\spn(V_S)$ the span of the columns of $V$ indexed by $S$.
The next two lemmas are crucial to the analysis.

\begin{lemma}\label{lem:key-intersection}
Let $X$ be a matrix as above. Let $A, B \subseteq [R]$, with $|B|=q$ and $A \cap B = \emptyset$. For $1 \le i \le q$, define $T_i$ to be the union of $A$ with all elements of $B$ except the $i$th one (when indexed in some way). Suppose further that $|A|+| B| \le k$. Then if $y \in \R^n$ is $\eps$-close to $\spn(X_{T_i})$ for each $i$, it is in fact $\errp_{\ref{lem:key-intersection}} \cdot \eps$ close $\spn(X_A)$, where $\errp_{\ref{lem:key-intersection}} := 4n\ka \rho$.
\end{lemma}
\begin{proof}
W.l.o.g., let us suppose $B = \{1, \dots, q\}$. Also, let $x_j$ denote the $j$th column of $X$.  From the hypothesis, we can write:
\begin{align*}
y = u_1  &+ \sum_{j \ne 1} \alpha_{1j} x_j + z_1\\
y = u_2  &+ \sum_{j \ne 2} \alpha_{2j} x_j + z_2\\
&\vdots \\
y = u_q  &+ \sum_{j \ne q} \alpha_{rj} x_j + z_q,
\end{align*}
where $u_i \in \spn(X_A)$ and $z_i$ are the {\em error} vectors, which by hypothesis satisfy $\norm{z_i}_2 < \eps$. We will use the fact that $|A|+|B| \le k$ to conclude that {\em each} $\alpha_{ij}$ is tiny. This then implies the desired conclusion.

By equating the first and $i$th equations ($i\ge 2$), we obtain
\[ u_1 + \sum_{j \ne 1} \alpha_{1j} x_j + z_1 = u_i + \sum_{j \ne 1} \alpha_{ij} x_j + z_i.\]
Thus we have a combination of the vectors $x_i$ being equal to $z_i-z_1$, which by hypothesis is small: $\norm{z_i-z_1}_2 \le 2\eps$. Now the key is to observe that the coefficient of $x_i$ is precisely $\alpha_{1i}$, because it is zero in the $i$th equation. Thus by Lemma~\ref{lem:conditioned1} (since $\krk{\ka}(X) \ge k$), we have that $|\alpha_{1i}| \le 2\ka \eps$.

Since we have this for all $i$, we can use the first equation to conclude that
\[ \norm{y-u_1}_2 \le  \sum_{j \ne 1} |\alpha_{1j}| \norm{x_j}_2 + \norm{z_1}_2 \le 2q\ka \rho \eps + \eps < 4n \ka \rho \eps \]

The last inequality is because $q < n$, and this completes the proof.
\end{proof}

A counting argument lies at the core of the inductive proof. We present it in terms of sunflower set systems, since it allows for a clean presentation.

\begin{definition}[Sunflower set system]
A set system $\calF$ is said to be a ``sunflower on $[R]$ with core $T^*$'' if $\calF \subseteq 2^{[R]}$, and for any $F_1, F_2 \in \cal{F}$, we have $F_1 \cap F_2 \in T^*$.
\end{definition}

\begin{lemma}\label{lem:sunflower}
Let $\{T_1, T_2, \dots, T_q\}$, $q\ge 2$, be a sunflower on $[R]$ with core $T^*$, and suppose $|T_1| + |T_2| + \dots + |T_q| \ge R + (q-1)\theta$, for some $\theta$. Then we have $|T^*| \ge \theta$, and furthermore, equality occurs iff $T^* \subseteq T_i$ for all $1 \le i \le q$.
\end{lemma}
\begin{proof}
The proof is by a counting argument. By the sunflower structure, each $T_i$ has some intersection with $T^*$, and some elements which do not belong to $T_{i'}$ for any $i' \ne i$. Call the number of elements of the latter kind $t_i$. Then we must have
\[ R + (q-1) \theta \le \sum_i |T_i| = \sum_i (t_i + |T_i \cap T^*|) \le \sum_i t_i + q|T^*|.\]
Now since all $T_i \subseteq [R]$, we have
\[ \sum_i t_i + |T^*| \le R.\]
Combining the two, we obtain
\[ R + (q-1) \theta \le R + (q-1) |T^*| ~\implies ~ |T^*| \ge \theta,\]
as desired. For equality to occur, we must have equality in each of the places above, in particular, we must have $|T_i \cap T^*| = |T^*|$ for all $i$, which implies $T^* \subseteq T_i$ for all $i$.
\end{proof}

Finally, we introduce a bit more notation before getting to the proof. For $S \subseteq [R]$ of size $(k-1)$, we define $T_S$ to be the set of indices corresponding to columns of $Y$ which are $\eps_1$-close to $\spn(X_S)$, where $\eps_1 := (nR) \eps$, and $\eps$ is as defined in the statement of Lemma~\ref{lem:permutation-poly}. For smaller sets $S$, we define:
\[ T_S := \bigcap_{|S'| = (k-1), S' \supset S} T_{S'}. \]

With the above lemmas in place, we can prove Lemma~\ref{lem:permutation-poly}.
\begin{proof}[Proof of Lemma~\ref{lem:permutation-poly}.]
We first prove the following claim by induction:
\begin{quote}
{\em Claim.} For every $S \subseteq [R]$ of size $\le (k-1)$, we have $|T_S| = |S|$.
\end{quote}
We do this by downward induction on $|S|$. For $|S| = k-1$, the hypothesis of the theorem implies that $|T_S| \ge k-1$. To see this, let $V$ be the $(n-k+1)$ dimensional space orthogonal to the span of $X_S$, and let $t$ be the number of columns of $Y$ which have a projection $>\eps_1$ onto $V$. From Lemma~\ref{lem:perm:basecase} (applied to the projections to $V$), there is a unit vector $w \in V$ with dot-product of magnitude $> \eps_1/Rn = \eps$ with each of the $t$ columns. From the hypothesis, since $w \in V$ ($\implies nz(w^T X) \le R-k+1$), we have $t \le R-k+1$. Thus at least $(k-1)$ of the columns are $\eps_1$-close to $\spn(X_S)$. Now since $\krk{\ka} (Y) \ge k$, it follows that $k$ columns of $Y$ cannot be $\eps_1$-close the $(k-1)$-dimensional space $\spn(X_S)$ (Lemma~\ref{lem:conditioned:span}). Thus $|T_S| = k-1$.

Now consider some $S$ of size $|S| \le k-2$. W.l.o.g., we may suppose it is $\{R-|S|+1, \dots, R\}$. Let $W_i$ denote $T_{S \cup \{i\}}$, for $1 \le i \le R-|S|$, and let us write $q = R-|S|$. By the inductive hypothesis, $|W_i| \ge |S|+1$ for all $i$.

Let us define $T^*$ to be the set of indices of the columns of $Y$ which are $\eps_1 \cdot \errp_{\ref{lem:key-intersection}}$-close to $\spn(X_S)$. We claim that $W_i \cap W_j \subseteq T^*$ for any $i \ne j$. This can be seen as follows: first note that $W_i \cap W_j$ is contained in the intersection of $T_{S'}$, where the intersection is over $S'$ such that $|S'| = k-1$, and $S'$ contains either $i$ or $j$. Now consider any $k-|S|$ element set $B$ which contains both $i,j$ (note $|S| \le k-2$). The intersection above includes sets which contain $S$ along with all of $B$ except the $r$th element (indexed arbitrarily), for each $r$. Thus by Lemma~\ref{lem:key-intersection}, we have that $W_i \cap W_j \subseteq T^*$.

Thus the sets $\{W_1, \dots, W_q\}$ form a sunflower family with core $T^*$. Further, we can check that the condition of Lemma~\ref{lem:sunflower} holds with $\theta = |S|$: since $|W_j| \ge |S|+1$ by the inductive hypothesis, it suffices to verify that
\[ R + (q-1) |S| \le q (|S|+1), \text{ which is true since } R = q + |S|.\]
Thus we must have $|T^*| \ge |S|$.

But now, note that $T^*$ is defined as the columns of $Y$ which are $\eps_1 \cdot \errp_{\ref{lem:key-intersection}}$-close to $\spn(X_S)$, and thus $|T^*| \le |S|$ (by Lemma~\ref{lem:conditioned:span}), and thus we have $|T^*| = |S|$. Now we have equality in Lemma~\ref{lem:sunflower}, and so the `furthermore' part of the lemma implies that $T^* \subseteq W_i$ for all $i$. 

Thus we have $T_S = \bigcap_i W_i = T^*$ (the first equality follows from the definition of $T_S$), thus completing the proof of the claim, by induction. 

Once we have the claim, the theorem follows by applying to singleton sets. Let $S = \{i\}$. Now if $y$ is a column of $Y$ which is in $\spn(X_{S'})$ for all $(k-1)$ element subsets $S'$ (of $[R]$) which contain $i$, by Lemma~\ref{lem:key-intersection}, we have $y$ being $\eps_1 \cdot \errp_{\ref{lem:key-intersection}}$-close to $\spn(X_{\{i\}})$, which implies $\norm{y - \alpha x_i}_2 \le \eps_1 \cdot \errp_{\ref{lem:key-intersection}}$. Since this is true for each column $i$, and since $k \ge 2$ the lemma follows.
\end{proof}

%% file: robustkruskal-higher.tex
\subsection{Uniqueness Theorem for Higher Order Tensors}\label{sec:unique-higher}

We show the uniqueness theorem for higher order tensors by a reduction to third order tensors as in \cite{SB00}. This reduction will proceed inductively, i.e., the robust uniqueness of order $\ell$ tensors is deduced from that of order $(\ell-1)$ tensors. We will convert an order $\ell$ tensor to a order $(\ell-1)$ tensor by combining two of the components together (say last two) as a $n_{\ell-1}n_\ell$ dimensional vector ($\spc{U}{\ell-1} \otimes \spc{U}{\ell}$ say). This is precisely captured by the Khatri-Rao product of two matrices:

\begin{definition}[Khatri-Rao product]\label{def:khatri}
Given two matrices $A$ (size $n_1 \times R$) and $B$ (size $n_2 \times R$),
the $(n_1 n_2)\times R$ matrix $M=A \odot B$ constructed with the $i^{th}$ column equal to $M_i =A_i \otimes B_i$ (viewed as a vector) is the Khatri-Rao product.
\end{definition}

Lemma~\ref{lem:krprod} in the appendix relates the \kr of $A \odot B$ with $k_A =\krk{\ka_1}(A)$ and $k_B = \krk{\ka_2}(B)$. It shows that $\krk{\ka_1 \ka_2 \errp} (A \odot B)=\min\{k_A+k_B-1,R\}$, for some $\errp$. This turns out to be crucial to the proof of uniqueness in the general case, which we present now.

\noindent {\bf Outline.} The proof proceeds by induction on $\ell$. The base case is $\ell=3$, and for higher $\ell$, the idea is to reduce to the case of $\ell-1$ by taking the Khatri-Rao product of the vectors in two of the dimensions. That is, if $[\spc{U}{1} ~ \spc{U}{2} ~ \dots ~ \spc{U}{\ell}]$ and $[\spc{V}{1} ~ \spc{V}{2} ~ \dots ~ \spc{V}{\ell}]$ are close, we conclude that $[\spc{U}{1} ~ \spc{U}{2} ~ \dots ~ (\spc{U}{\ell-1} \odot \spc{U}{\ell})]$ and $[\spc{V}{1} ~ \spc{V}{2} ~ \dots ~ (\spc{V}{\ell-1} \odot \spc{V}{\ell})]$ are close, and use the inductive hypothesis, which holds because of Lemma~\ref{lem:krprod} we mentioned above. We then need an additional step to conclude that if $A \odot B$ and $C \odot D$ are close, then so are $A, C$ and $B, D$ up to some loss (Lemma~\ref{lem:tensoring} -- this is where we have a {\em square root} loss, which is why we have a bad dependence on the $\eps'$ in the statement). We now formalize this outline.

\begin{proof}[Proof of Theorem~\ref{thm:unique:gen}]
We will prove by induction on $\ell$. The base case of $\ell=3$ is established by Theorem~\ref{thm:unique3}. Thus consider some $\ell \ge 4$, and suppose the theorem is true for $\ell-1$. Furthermore, suppose the parameters $\eps$ and $\eps'$ in the statement of Theorem~\ref{thm:unique:gen} for $(\ell-1)$ be $\eps_{\ell-1}$ and $\eps'_{\ell-1}$. We will use these to define $\eps_\ell$ and $\eps'_\ell$ which correspond to parameters in the statement for $\ell$.

Now consider $U^{(i)}$ and $V^{(i)}$  as in the statement of the theorem. Let us assume without loss of generality that $k_1 \ge k_2 \ge \dots \ge k_\ell$. Also let $K=\sum_{j \in [\ell]} k_j$.
We will now combine the last two components $(\ell-1)$ and $\ell$ by the Khatri-Rao product. 
\[ \tU=\spc{U}{\ell-1} \odot \spc{U}{\ell} ~\text{and } ~ \tV=\spc{V}{\ell-1} \odot \spc{V}{\ell}. \]

Since we know that the two representations are close in Frobenius norm, we have
\begin{equation}
\nrm{\sum_{r \in [R]} \spc{U}{1}_r \otimes \spc{U}{2}_r\otimes \dots \otimes \spc{U}{\ell-2}_r \otimes \tU_r- \sum_{r \in [R]} \spc{V}{1}_r \otimes \spc{V}{2}_r \otimes \dots \otimes \spc{V}{\ell-2}_r \otimes \tV_r}{F} < \eps_\ell
\end{equation}
Let us first check that the conditions for $(\ell-1)$-order tensors hold for $\tka=(\ka_{\ell-1} \ka_\ell  \sqrt{K}) \le (\ka_{\ell-1}\ka_{\ell} \sqrt{3R})$. From Lemma~\ref{lem:krprod}, $\krk{\tka}(\tU) \ge \min\{k_\ell+k_{\ell-1}-1, R\}$.

Suppose first that $k_{\ell}+k_{\ell-1} \le R+1$, then 
\[ \sum_{j\in [\ell-1]}k'_j \ge \sum_{j \in [\ell-2]} k_j + k_{\ell-1}+k_{\ell} -1 \ge 2R+(\ell-1)-1. \]
Otherwise, if $k_{\ell}+k_{\ell-1} > R+1$, then $k_{\ell-3}+k_{\ell-2} \ge R+2$ (due to our ordering, and $\ell \ge 4$). Hence  \[ \sum_{j \in [\ell-1]} k'_j \ge (\ell-4)+(R+2)+(R+1) \ge 2R+\ell-1 \]

We now apply the inductive hypothesis on this $(\ell-1)$th order tensor. 
Note that $\tilde{\rho}\le (\rho_{\ell-1}\rho_\ell)$, $\tilde{\rho'}\le (\rho'_{\ell-1} \rho'_\ell)$, $\tka \le \left(2\ka_{\ell-1} \ka_\ell \sqrt{R}\right)$ and $\tilde{n}=n_{\ell-1} n_\ell$.

We will in fact apply it with $\eps'_{\ell-1} < \min\{ (R \cdot \ka_{\ell-1}\ka_{\ell} \cdot \rho'_{\ell-1}\rho'_\ell )^{-2}, ~  (\eps'_\ell)^2/ R \}$, so that we can later use Lemma~\ref{lem:tensoring}. To ensure these, we will set
\begin{align*}
\eps_{\ell}^{-1} &=  \errp_{\ref{thm:unique:gen}}^{\ell}\left(\frac{R}{\eps'_\ell}\right) \cdot \left(\prod_{j \in [\ell-2]} \errp_{\ref{thm:unique:gen}}(\ka_j,\rho_j,\rho'_j,n_j)\right) \errp_{\ref{thm:unique:gen}}(\tka,\tilde{\rho},\tilde{\rho'},\tilde{n}),
\end{align*}
where $\errp_{\ref{thm:unique:gen}}^{\ell} = x^{O(2^\ell)}$. From the values of $\tka, \tilde{\rho}, \tilde{n}$ above, this can easily be seen to be of the form in the statement of the theorem.


The inductive hypothesis implies that there is a permutation matrix $\Pi$ and scalar matrices $\{\spc{\Lambda}{1},\spc{\Lambda}{2},\dots,\spc{\Lambda}{\ell-2},\Lambda' \}$, such that
$\norm{\spc{\Lambda}{1}\spc{\Lambda}{2}\dots\spc{\Lambda}{\ell-2}\Lambda' -I} < \eps_{\ell-1}'$
and 
\begin{align*}
\forall j \in [\ell-2] \quad \nrm{V^{(j)}- U^{(j)}\Pi \spc{\Lambda}{j}}{F} &< \eps_{\ell-1}' \\
\nrm{\tV- \tU\Pi \Lambda'}{F} &< \eps_{\ell-1}'
\end{align*} 

\noindent Since $\eps_{\ell-1}' < \eps_\ell'$, equation \eqref{eq:uniquegen:1} is satisfied for $j\in [\ell-2]$.
We thus need to show that $\nrm{V^{(j)}- U^{(j)}\Pi \spc{\Lambda}{j}}{F}<\eps'_\ell$ for $j=\ell-1$ and $\ell$.
To do this, we appeal to Lemma~\ref{lem:tensoring}, to say that if the Frobenius norm of the difference of two tensor products $u \otimes v$ and $u' \otimes v'$ is small, then the component vectors are nearly parallel. 

Let us first set the parameters for applying Lemma~\ref{lem:tensoring}. 
Each column vector is of length at most $\maxl \le \tilde{\rho'} \le (\rho'_{\ell-1}\rho'_\ell)$ and length at least $\minl \ge 1/\tka \ge (2\ka_{\ell-1}\ka_\ell \sqrt{R})$. Hence, because of our choice of $\eps_{\ell-1}' \ll \left(4 \sqrt{R}(\ka_{\ell-1}\ka_{\ell})(\rho'_{\ell-1} \rho'_\ell)\right)^{-1}$ earlier, the conditions of Lemma~\ref{lem:tensoring} are satisfied with $\delta \le \eps_\ell'$. Let $\delta_r = \nrm{\tV_r-\tU_{\pi(r)} \Lambda'(r)}{2}$.

Now applying Lemma~\ref{lem:tensoring} with $\delta=\delta_r$, to column $r$, we see that there are scalars $\alpha_r(\ell-1)$ and $\alpha_r(\ell)$ such that 
$$\abs{1-\alpha_r(\ell-1)\alpha_r(\ell)} < \frac{\eps'_{\ell-1}}{\minl^2} \le \eps'_\ell. $$ 
By setting for all $r\in [R]$, $\spc{\Lambda}{\ell-1}(r)=\alpha(\ell-1)_r$ and $\spc{\Lambda}{\ell}(r)=\alpha(2) \Lambda'(r)$, we see that the first part of \eqref{eq:uniquegen:1} is satisfied.
Finally, Lemma~\ref{lem:tensoring} shows that 
\begin{align*}
\forall j \in \{\ell-1,\ell\} ~& ~ \nrm{\spc{V}{j}_r-\spc{U}{j}_{\pi(r)}\spc{\Lambda}{j}(r)}{2} < \sqrt{\delta_r}~, \quad \forall r \in [R]\\
& ~~~\nrm{V^{(j)}-U^{(j)} \Pi \spc{\Lambda}{j}}{F} < R^{1/4} \sqrt{\eps'_{\ell-1}} \quad \text{( by Cauchy-Schwartz inequality).}\\
&\qquad \qquad \qquad \qquad \qquad<\eps'_\ell
\end{align*}
This completes the proof of the theorem.
\end{proof}


We show a similar result for symmetric tensors, which shows robust uniqueness upto permutations (and no scaling) which will be useful in applications to mixture models (Section~\ref{sec:applications}).

\begin{corollary}[Unique Symmetric Decompositions]\label{corr:unique:sym}
For every $0<\eta<1$, $\ka,\rho,\rho'>0$ and $\ell,R \in \bbN$, 
$\exists \eps_\ell = \vartheta^{(\ell)}_{\ref{corr:unique:sym}}(\frac{1}{\eta},R,n,\ka,\rho,\rho')$ such that, 
for any $\ell$-order symmetric tensor (with $\ell \le R$) $$T=\sum_{r \in [R]} \bigotimes_{j=1}^{\ell} U_r$$ 
where the matrix $U$ is $\rho$-bounded with $\krk{\ka}(U)=k \ge \frac{2R-1}{\ell}+1$, and 
for any other $\rho'$ bounded, symmetric, rank-$R$ decomposition of $T$ which is $\eps$-close, i.e., 
$$\nrm{\sum_{r \in [R]} \bigotimes_{j=1}^{\ell} V_r - \sum_{r \in [R]} \bigotimes_{j=1}^{\ell} U_r}{F} \le \eps$$ 
there exists an $R \times R$ permutation matrix $\Pi$ such that 
\begin{equation}
\nrm{V - U \Pi}{F} \le \eta 
\end{equation}
\end{corollary}

The mild intricacy here is that applying Theorem~\ref{thm:unique:gen} gives a bunch of scalar matrices whose product is close to the identity, while we want each of the matrices to be so. This turns out to be easy to argue -- see Section~\ref{sec:unique:sym}.

\newcommand{\Rr}{\mathbb{R}}


%% file: tensordecomp.tex
For matrices, the theory of low rank approximation is well understood, and they are captured using singular values. In contrast, the tensor analog of the problem is in general ill-posed: for instance, there exist rank-3 tensors with arbitrarily good rank $2$ approximations~\cite{Lan}. For instance if $u, v$ are orthogonal vectors, we have
\[ u \opr v \opr v + v \opr u \opr v + v \opr v \opr u = \frac{1}{\eps} \big[ (v+\eps u)\opr (v+\eps u ) \opr (v + \eps u) - v \opr v\opr v \big] + \calN, \]
where $\norm{\calN}_F \le O(\eps)$, while it is known that the LHS has rank $3$. However note that the rank-2 representation with error $\eps$ uses vectors of length $1/\eps$, and such {\em cancellations}, in a sense are responsible for the ill-posedness.

Hence in order to make the problem well-posed, we will impose a boundedness assumption.
\begin{definition}[$\rho$-bounded Low-rank Approximation]
Suppose we are given a parameter $R$ and an $m \times n \times p$ tensor $T$ which can be written as
\begin{equation}\label{eqn:lowrankrep}
T = \sum_{i = 1}^R a_i \opr b_i \opr c_i + \calN,
\end{equation}
where $a_i \in \R^m, b_i \in \R^n, c_i \in \R^p$ satisfy $\max\{\norm{a_i}_2, \norm{b_i}_2, \norm{c_i}_2 \} \le \rho$, and $\calN$ is a {\em noise} tensor which satisfies $\norm{\calN}_F \le \eps$, for some small enough $\eps$. The $\rho$-bounded low-rank decomposition problem asks to recover a {\em good} low rank approximation, i.e.,
\[ T = \sum_{i = 1}^R a_i' \opr b_i' \opr c_i' + \calN', \]
such that $a_i', b_i', c_i'$ are vectors with norm at most $\rho$, and $\norm{\calN'}_F \le O(1) \cdot \eps$.
\end{definition}

We note that if the decomposition into $\tens{A}{B}{C}$ above satisfies the conditions of Theorem~\ref{thm:unique3}, then solving the $\rho$-bounded low-rank approximation problem would allow us to recover $A, B, C$ up to a small error. The algorithmic result we prove is the following (restated version of Theorem~\ref{thm:algorithm}).
\begin{theorem}
The $\rho$-bounded low-rank approximation problem can be solved in time $\poly(n) \cdot \exp(R^2 \log(R\rho /\eps))$.
\end{theorem}
In fact, the $O(1)$ term in the error bound $\calN' \le O(1) \cdot \eps$ will just be $5$. Our algorithm is extremely simple conceptually: we identify three $R$-dimensional spaces by computing appropriate SVDs, and prove that for the purpose of obtaining an approximation with $O(\eps)$ error, it suffices to look for $a_i, b_i, c_i$ in these spaces. We then find the approximate decomposition by a brute force search using an epsilon-net.
Note that the algorithm has a polynomial running time for constant $R$, which is typically when the low rank approximation problem is interesting.

\begin{proof}
In what follows, let $M_A$ denote the $m \times np$ matrix whose columns are the so-called $j,k$th {\em modes} of the tensor $T$, i.e., the $m$ dimensional vector of $T_{ijk}$ values obtained by fixing $j,k$ and varying $i$. Similarly, we define $M_B ~(n \times mp)$ and $M_C ~(p \times mn)$. Also, we denote by $A$ the $m \times R$ matrix with columns being $a_i$. Similarly define $B ~(n \times R), C ~(p \times R)$.

The outline of the proof is as follows: we first observe that the matrices $M_A, M_B, M_C$ are all approximately rank $R$. We then let $V_A, V_B$ and $V_C$ be the span of the top $R$ singular vectors of $M_A, M_B$ and $M_C$ respectively, and show that it suffices to search for $a_i, b_i$, and $c_i$ in these spans. We note that we do not (and in fact cannot, as simple examples show) obtain the {\em true} span of the $a_i$, $b_i$ and $c_i$'s in general. Our proof carefully gets around this point. We then construct an $\eps$-net for $V_A, V_B, V_C$, and try out all possible $R$-tuples. This gives the roughly $\exp(R^2)$ running time claimed in the Theorem.

We now make formal claims following the outline above.
\begin{claim}
Let $V_A$ be the span of the top $R$ singular vectors of $M_A$, and let $\Pi_A$ be the projection matrix onto $V_A$ (i.e., $\Pi_A v$ is the projection of $v\in \R^n$ onto $V_A$). Then we have
\[ \norm{ M_A - \Pi_A M_A }_F \le \eps \]
\end{claim}
\begin{proof}
Because the top $R$ singular vectors give the best possible rank-$R$ approximation of a matrix for every $R$, for any $R$-dimensional subspace $S$, if $\Pi_S$ is the projection matrix onto $S$, we have
\[ \norm{M_A - \Pi_A M_A }_F \le \norm{M_A - \Pi_S M_A }_F \]
Picking $S$ to be the span of the vectors $\{ a_1, \dots, a_R \}$, we obtain 
\[ \norm{M_A - \Pi_S M_A }_F \le \norm{\calN}_F \le \eps. \]
The first inequality above is because the $j,k$th mode of the tensor $\sum_i a_i \opr b_i \opr c_i$ is a vector in the span of $\{a_1, \dots, a_R\}$, in particular, it is equal to $\sum_i b_i(j) c_i(k) a_i$, where $b_i(j)$ denotes the $j$th coordinate of $b_i$.

This completes the proof.
\end{proof}

\newcommand{\atil}{\widetilde{a}}
\newcommand{\btil}{\widetilde{b}}
\newcommand{\ctil}{\widetilde{c}}
\newcommand{\acap}{\widehat{a}}
\newcommand{\bcap}{\widehat{b}}
\newcommand{\ccap}{\widehat{c}}
\newcommand{\abar}{a^{\perp}}
\newcommand{\bbar}{b^{\perp}}
\newcommand{\cbar}{c^{\perp}}

Next, we will show that looking for $a_i, b_i, c_i$ in the spaces $V_A, V_B, V_C$ is sufficient. The natural choices are $\Pi_A a_i, \Pi_B b_i, \Pi_C c_i$, and we show that this choice in fact gives a good approximation. For convenience let $\atil_i := \Pi_A a_i$, and $\abar_i := a_i - \atil_i$.

\begin{claim}\label{lem:decomp:claim}
For $T, V_A, \atil_i, \dots$ as defined above, we have
\[ \norm{T - \calN - \sum_i \atil_i \opr  \btil_i \opr  \ctil_i}_F \le 3\eps. \]
\end{claim}
\begin{proof}
The proof is by a {\em hybrid argument}. We write
\begin{align*}
T - \calN - \sum_i \atil_i \opr \btil_i \opr \ctil_i &= \big( \sum_i a_i \opr b_i \opr c_i - \atil_i \opr b_i \opr c_i \big) \\
&+ \big( \sum_i \atil_i \opr b_i \opr c_i - \atil_i \opr \btil_i \opr c_i \big) \\
&+ \big( \sum_i \atil_i \opr \btil_i \opr c_i - \atil_i \opr \btil_i \opr \ctil_i \big).
\end{align*}

We now bound each of the terms in the parentheses, and then appeal to triangle inequality (for the Frobenius norm). Now, the first term is easy:
\[ \norm{\sum_i a_i \opr b_i \opr c_i - \atil_i \opr b_i \opr c_i}_F = \norm{M_A - \Pi_A M_A}_F \le \eps. \]

One way to bound the second term is as follows. Note that:
\[ \sum_i a_i \opr b_i \opr c_i - a_i \opr \btil_i \opr c_i = \Big( \sum_i \atil_i \opr b_i \opr c_i - \atil_i \opr \btil_i \opr c_i \Big) + \Big( \sum_i \abar_i \opr b_i \opr c_i - \abar_i \opr \btil_i \opr c_i \Big).\]
Now let us denote the two terms in the parenthesis on the RHS by $G, H$ -- these are tensors which we view as $mnp$ dimensional vectors. We have $\norm{G+H}_2 \le \eps$, because the Frobenius norm of the LHS is precisely $\norm{M_B - \Pi_B M_B}_F \le \eps$. Furthermore, $\iprod{G,H} =0$, because $\iprod{\atil_i, \abar_j}=0$ for any $i,j$ (one vector lies in the span $V_A$ and the other orthogonal to it). Thus we have $\norm{G}_2 \le \eps$ (since in this case $\norm{G+H}_2^2 = \norm{G}_2^2 + \norm{H}_2^2$).

A very similar proof lets us conclude that the Frobenius norm of the third term is also $\le \eps$. This completes the proof of the claim, by our earlier observation.
\end{proof}

The claim above shows that there exist vectors $\atil_i, \btil_i, \ctil_i$ of length at most $\rho$ in $V_A, V_B, V_C$ resp., which give a rank-$R$ approximation with error at most $4\eps$. Now, we form an $\eps/(R\rho^2)$-net over the ball of radius $\rho$ in each of the spaces $V_A, V_B, V_C$. Since these spaces have dimension $R$, the nets have size
\[ \Big( \frac{O(R\rho^2)}{\eps} \Big)^R \le \exp(O(R) \log(R\rho/\eps)).\]

Thus let us try all possible candidates for $\atil_i, \btil_i, \ctil_i$ from these nets. Suppose we have $\acap_i, \bcap_i, \ccap_i$ being vectors which are $\eps/(6R\rho^2)$-close to $\atil_i, \btil_i, \ctil_i$ respectively, it is easy to see that
\begin{align*}
\norm{\sum_i \atil_i \opr \btil_i \opr \ctil_i - \acap_i \opr \bcap_i \opr \ccap_i }_F \le  \sum_i \norm{ \atil_i \opr \btil_i \opr \ctil_i - \acap_i \opr \bcap_i \opr \ccap_i }_F
\end{align*}
Now by a hybrid argument exactly as above, and using the fact that all the vectors involved are $\le \rho$ in length, we obtain that the LHS above is at most $\eps$. 

Thus the algorithm finds vectors such that the error is at most $5\eps$. The running time depends on the time taken to try all possible candidates for $3R$ vectors, and evaluating the tensor for each. Thus it is $\poly(m,n,p) \cdot \exp(O(R^2) \log(R\rho/\eps))$.
\end{proof}

This argument generalizes in an obvious way to order $\ell$ tensors,
and gives the following. We omit the proof.
\begin{theorem}\label{thm:algo}
There is an algorithm, that when given an order $\ell$ tensor of size $n$ with a rank $R$ approximation of error $\eps$ (in $\nrm{\cdot}{F}$), finds a rank-$R$ approximation of error $O(\ell \eps)$ in time $\poly(n) \cdot \exp(O(\ell R^2) \log(\ell R\rho /\eps))$.
\end{theorem}

%% file: applications.tex


We now show how our robust uniqueness theorems for tensor decompositions can be used for learning latent variable models, with polynomial sample complexity bounds. 

\begin{definition}[Polynomial Identifiability]
An instance of a hidden variable model of size $m$ with hidden variables set $\Upsilon$ is said to be polynomial identifiable if there is an algorithm that given any $\eta>0$, uses only $N \le \poly(m,1/\eta)$ samples and finds with probability $1-o(1)$ estimates of the hidden variables $\Upsilon'$ such that $\nrm{\Upsilon' -  \Upsilon}{\infty}< \eta$.
\end{definition}

Consider a simple mixture-model, where each sample is generated from mixture of $R$ distributions $\{\calD_r\}_{r \in [R]}$, with mixing probabilities $\{w_r\}_{r \in [R]}$. Here the latent variable $h$ corresponds to the choice of distribution and it can have $[R]$ possibilities. First the distribution $h=r$ is picked with probability $w_r$, and then the data is sampled according to $\calD_r$, which has mean $\mu_r \in \R^n$. Let $M_{n \times R}$ represent the matrix of these $R$ means. The goal is to learn these hidden parameters ($M$ and weights $\{w_r\}$) after observing many samples. This setting captures many latent variable models including topic models, HMMs, gaussian mixtures etc. 
  
While practitioners typically use Expectation-Maximization (EM) methods to learn the parameters, a good alternative in the case of mixture models is using \emph{the method of moments} approach ( starting from the work by Pearson \cite{Pea94} for univariate gaussians ), which tries to identify the parameters by estimating higher order moments. However, one drawback is that the number of moments required is typically as large as the number of mixtures $R$ (or parameters), resulting in a sample complexity that is exponential in $R$ \cite{MV10,BS10,FOS05,FOS06}.

In a recent exciting line of work \cite{MR06,AHK12,HK12,AFHKL12,AGHKT12}, it is shown that $\poly(R,n)$ samples suffice for identifiability in a special case called the \emph{non-singular} or \emph{non-degenerate} case i.e. when the matrix $M$ has full rank (rank = $R$)\footnote{For polynomial identifiability, $\sigma_R \ge 1/\poly(n)$.} for many of these models. Their algorithms for this case proceed by reducing the problem of finding the latent variables (the means and weights) to the problem of decomposing {\em Symmetric Orthogonal Tensors} of order $3$, which are known to be solvable in $\poly(n,R)$ time using power-iteration type methods \cite{KR01,ZG01,AGHKT12}. 

However, their approach crucially relies on these non-degeneracy conditions, and are not robust: even in the case when these $R$-means reside in a $(R-1)$-dimensional space, these algorithms fail, and the best known sample complexity bounds in many of these settings are $\exp(R)\poly(n)$. In many settings like speech recognition and image classification, the dimension of the feature space $n$ is typically much smaller than $R$, the number of topics or clusters. For instance, the (effective) feature space corresponds to just the low-frequency components in the fourier spectrum for speech, or the local neighborhood of a pixel in images (SIFT features \cite{Low99}). These are typically much smaller than the different kinds of objects or patterns (topics) that are possible. Further, in other settings, the set of relevant features (the effective feature space) could be a space of much smaller dimension ($k < R$) that is unknown to us even when the feature vectors are actually represented in a large dimensional space ($n\gg R$). 

In this section, we show that we can use our Robust Uniqueness results for Tensor Decompositions (Theorem~\ref{thm:unique3} and Theorem~\ref{thm:unique:gen}) to go past the non-degeneracy barrier and prove that $\poly(R,n)$ samples suffice even under the milder condition that no $k=\delta R$ gaussians lie in a $(k-1)$ dimensional space (for some constant $\delta>0$). Further, these results generalize to other hidden variable models like Topic Modeling, Hidden Markov models, Mixture models etc. One interesting aspect of our approach is that, unlike previous works, we get a smooth tradeoff : we get polynomial identifiability under successively milder conditions by using higher order tensors ($\ell \approx 2/\delta$). This reinforces the intuition that higher moments capture more information at the cost of efficiency. 

In the rest of this section, we will first describe Multi-view models and show how the robust uniqueness theorems for tensor decompositions imply polynomial identifiability in this model. We will then see two popular latent variable models which fit into the multi-view mixture model: the exchangeable (single) Topic Model and Hidden Markov models. 
We note that the results of this section (for $\ell=3$ views) also apply to other latent variable models like \emph{Latent Dirichlet Allocation (LDA)} and \emph{Independent Component Analysis (ICA)} that were studied in \cite{AGHKT12}. 
We omit the details in this version of the paper.

\subsection{Multi-view Mixture Model}

Multi-view models are mixture models with a latent variable $h$, where we are given multiple observations or views $\spc{x}{1}, \spc{x}{2}, \dots,\spc{x}{\ell}$ that are conditionally independent given the latent variable $h$. Multi-view models are very expressive, and capture many well-studied models like Topic Models \cite{AHK12}, Hidden Markov Models (HMMs) \cite{MR06,AMR09,AHK12}, random graph mixtures \cite{AMR09}.
We first introduce some notation, along the lines of \cite{AMR09, AHK12}. 

\begin{definition}[Multi-view mixture models]\label{def:MM}
~ \\
\begin{itemize}
\item The latent variable $h$ is a discrete random variable having domain $[R]$, so that $\Prb{h= r}= w_r,  \forall r \in [R]$.  
\item The views $\{\spc{x}{j}\}_{j \in [\ell]}$ are random vectors $\in \R^n$ that are conditionally independent given $h$, with means $\spc{\mu}{j} \in \R^n$ i.e.
$$\E{\spc{x}{j} \vert h=r}= \spc{\mu}{j}_r \text{ and } \E{\spc{x}{i} \otimes \spc{x}{j}\vert h=r}= \spc{\mu}{i}_r \otimes \spc{\mu}{j}_r \text{ for } i \ne j$$
\item Denote by $\spc{M}{j}$, the $n \times R$ matrix with the means $\{\spc{\mu}{j}_r\}_{r \in [R]}$ comprising its columns i.e. $$\spc{M}{j}=[ \spc{\mu}{j}_1 | \dots | \spc{\mu}{j}_r | \dots | \spc{\mu}{j}_R ].$$
\item The entries (domain) of $\spc{x}{j}$ are bounded by $\Bd$ i.e. $\nrm{\spc{x}{j}}{\infty} \le \Bd$.
\footnote{in general, we can also allow them to be continuous distributions like multivariate gaussians.} 
\end{itemize}
\end{definition}
The parameters of the model to be learned are the matrices $\{\spc{M}{j}\}_{j \in [\ell]}$ and the mixing weights $\{w_r\}_{r \in [R]}$. In many settings, the $n$-dimensional vectors $\spc{x}{j}$ are actually indicator vectors (hence $\Bd=1$): this is commonly used to encode the case when the observation is one of $n$ discrete events. Allman et al \cite{AMR09} refer to these models by {\em finite mixtures of finite measure products}. 

The following lemma shows how to obtain a higher order tensor (to apply our results from previous sections) in terms of the hidden parameters that we need to recover. It follows easily because of conditional independence.
\begin{lemma}[\cite{AMR09,AHK12}]
In the notation established above for multi-view models, $\forall \ell \in \bbN$ the $\ell^{th}$ moment tensor
$$\E{\spc{x}{1} \otimes \dots \spc{x}{j} \otimes \dots \spc{x}{\ell}} = \sum_{r \in [R]} w_r \spc{\mu}{1}_r \otimes \spc{\mu}{2}_r \dots \otimes \spc{\mu}{j}_r \otimes \dots \otimes \spc{\mu}{\ell}_r.$$
In our usual representation of tensor decompositions, 
$$ \E{\spc{x}{1} \otimes \dots \spc{x}{j} \otimes \dots \spc{x}{\ell}} = \left[\spc{M}{1} ~ \spc{M}{2} ~ \dots ~ \spc{M}{\ell} \right].$$
\end{lemma}

Recall that $\krk{\ka}(M)$ corresponds to the minimum number $k$ such that every $n \times k$ submatrix $M'$ of $M$ has $\sigma_k(M') > 1/\ka$. Intuitively this says that, no set of $k$ vectors from $\mu_{r \in [R]}$ all lie close to a $k-1$ dimensional space.

When $k \equiv \krk{\ka}(M) \ge R$ for each of these matrices (the non-degenerate or non-singular setting), Anandkumar et al. \cite{AHK12} give a polynomial time algorithm to learn the hidden variables using only $\poly(R,\ka,n)$ samples (hence polynomial identifiability). However, their algorithm fails even when $k = R-1$. We now how to achieve polynomial identifiability even when $k = \delta R$ for any constant $\delta>0$. 

\vspace{5pt}
\noindent {\bf Theorem~\ref{thm:MM}} (Polynomial identifiability of Multi-view mixture model). \emph{The following statement holds for any constant integer $\ell$. 
Suppose we are given samples from a multi-view mixture model (see Def~\ref{def:MM}), with the parameters satisfying:
\begin{enumerate}[(a)]
\item For each mixture $r \in [R]$, the mixture weight $w_r > \gamma$.
\item For each $j \in [\ell]$, $\krk{\ka}(\spc{M}{j})\ge k \ge \frac{2R}{\ell}+1$.
\end{enumerate}
then there is a algorithm that given any $\eta>0$ uses 
$N=\spc{\vartheta_{\ref{thm:MM}}}{\ell} \left(\frac{1}{\eta},R,n,\ka,1/\gamma,\Bd\right)$ samples, \\and finds with high probability $\{\spc{\tM}{j}\}_{j \in [\ell]} $ and $\{\tw_r\}_{r \in [R]}$ (upto renaming of the mixtures $\{1,2,\dots,R\}$) such that
\begin{equation}
\forall j \in [\ell], \quad \nrm{\spc{M}{j}-\spc{\tM}{j}}{F} \le \eta \quad \text{ and }\quad \forall r \in [R], ~ \abs{w_r - \tw_r} < \eta 
\end{equation}
Further, this algorithm runs in time $\exp\left(R^2 \ell^2 \left[2^{2\ell} \log(\frac{R\ell}{\eta})+\ell \log(n\cdot \frac{\ka \Bd}{\gamma}) \right] \right)\poly(n)$ time.}\\

\vspace{5pt}

Note that the above theorem shows polynomial identifiability (for constant $\ell$), and additionally gives an algorithm which takes time $n^{O_{\ell}(R^2)}\poly(\ka,R,\Bd)$ for inverse polynomial error. The function $\spc{\vartheta_{\ref{thm:MM}}}{\ell}(\cdot,\dots,\cdot)=\poly(Rn/(\gamma \eta))^{2^{\ell}} \poly(n,\ka,1/\gamma)^{\ell}$ is a polynomial for constant $\ell$ and satisfies the theorem.

\vspace{3pt}

\noindent\emph{ {\bf Remarks:} 
\begin{enumerate}
\item Note that the condition $(a)$ in the theorem about the mixing weights $w_r > \gamma$ is required to recover \emph{all} the parameters, since we need $\poly(1/w_r)$ samples before we see a sample from mixture $r$. However, by setting $\gamma \ll \eps'$, the above algorithm can still be used to recover the mixtures components of weight larger than $\eps'$. 
\item While these results give new polynomial sample complexity guarantees when $n < R$, they are interesting even when the dimension of the space $n \gg R$. A natural setting where this arises is when many of the vectors lie in a unknown space of much smaller dimension ($k$-dims), while the whole space has high dimension.
\item The theorem also holds when for different $j$, the $\krk{\tau}(\spc{M}{j})$ have bounds $k_j$ which are potentially different, and satisfy the same condition as in Theorem~\ref{thm:unique:gen}.
\end{enumerate}
}


\begin{proof}
We will consider the $\ell^{th}$ moment tensor for $\ell=\lceil 2/\delta \rceil+1$. 
The proof is simple, and proceeds in three steps. First, we use enough samples to obtain an estimate $\tilde{T}$ of the $\ell^{th}$ moment tensor $T$, upto inverse polynomial error. Then we find a good rank-$R$ approximation to $\tilde{T}$ (it exists because $T$ has rank $R$). We then use the Robust Uniqueness theorem for tensor decompositions to claim that the terms of this decomposition are in fact close to the hidden parameters. 

Set $\eta'=\frac{\eta \gamma}{16 \ell n}$. We know from Lemma~\ref{lem:samplingerror:topic} that the $\ell^{th}$ moment tensor can be estimated to accuracy \\
$\eps_1=\left(\ell \cdot \spc{\vartheta_{\ref{thm:unique:gen}}}{\ell}(R/\eta') \cdot \vartheta_{\ref{thm:unique:gen}}(\ka/\gamma,\Bd\sqrt{n},\Bd\sqrt{n},n)\right)^{-1} $ 
in $\|\cdot\|_F$ norm using $N=O(\eps_1^{-2}R(\Bd)^{\ell} \sqrt{\ell \log n})$ samples. This estimated tensor $\tilde{T}$ has a rank-$R$ decomposition upto error $\eps_1$. 

Next, we will apply our algorithm for getting approximate low-rank tensor decompositions from Section~\ref{sec:algo} on $\tilde{T}$. Since each $\spc{\mu}{j}_r$ is a probability distribution, we can obtain vectors $\{\spc{\tu}{j}_r\}_{j \in [\ell], r\in [R]}$ (let us call the corresponding $n \times R$  matrices $\spc{\tU}{j}$) such that 
$$\forall j \in [\ell-1], r \in [R] ~ \nrm{\spc{\tu}{j}_r}{1} \in [1-\delta,1+\delta] \quad \text{where } \delta=\eps_1 \sqrt{R} < \frac{\eta}{2\ell}.$$ 
This is possible since the algorithm in Section~\ref{sec:algo} searches for the vectors $\spc{\tu}{j}_r$, by just enumerating over $\eps$-nets on an $R$-dimensional space. An alternate way to see this is to obtain any decomposition and scale all but the last column in the matrices $\spc{\tU}{j}$ so that they have $\ell_1$ norm of $1$ (upto error $\delta$). Note that this step of finding an $\eps$-close rank-$R$ decomposition can also just comprise of brute force enumeration, if we are only concerned with polynomial identifiability. Hence, we have obtained a rank-$R$ decomposition which is $O(\ell \eps_1)$ far in $\|\cdot\|_F$. 

Now, we apply Theorem~\ref{thm:unique:gen} to $\ell^{th}$ moment tensor $T$ to claim that these $\spc{\tU}{j}$ are close to $\spc{M}{j}$ upto permutations. When we apply Theorem~\ref{thm:unique:gen}, we absorb the co-efficients $w_r$ into $\spc{M}{\ell}$. In other words
\[\spc{U}{j}=\spc{M}{j} ~\text{ for all } j \in [\ell-1], \quad \text{and} \quad \spc{U}{\ell}=\spc{M}{\ell}\diag(w).\]
We know that $\krk{\ka}(\spc{M}{j})=k_j$, and $\krk{\ka/\gamma}(\spc{U}{\ell})=k_\ell$.  
We now apply Theorem~\ref{thm:unique:gen} with our choice of $\eps_1$, and assuming that the permutation is identity without loss of generality,  we get 
\begin{align*}
\forall r \in [R] ~ & \norm{\spc{\tu}{j}_r-\spc{\Lambda}{j}(r)\spc{\mu}{j}_r} < \eta' \le \frac{\eta \gamma }{16 n \ell} \quad \forall j \in [\ell-1] \\
~\text{ and } ~ & \norm{\spc{\tu}{\ell}_r-\spc{\Lambda}{\ell}(r)w_r\spc{\mu}{\ell}_r} < \eta' \le \frac{\eta \gamma }{16\ell n}
\end{align*}
for some scalar matrices $\Lambda_j$ (on $R$-dims) such that 
\[\norm{\prod \spc{\Lambda}{j} -I_{R}} \le \frac{\eta}{16 \ell n} \label{eq:mvlambda}\]
Note that the entries in the diagonal matrices $\Lambda_j$ (the scalings) may be negative. We first transform the vectors so that each of the entries in $\Lambda_j$ are non-negative (this is possible since the product of $\Lambda_j$ is close to the identity matrix, which only has non-negative entries).
 
\begin{equation}
\forall j \in [\ell], r \in [R], \quad \spc{\tv}{j}_r = \sign\left(\spc{\Lambda}{j}(r)\right)\cdot \spc{\tu}{j}_r 
\end{equation}
This ensures that 
\begin{align}\label{eq:mvparams}
\forall j \in [\ell-1], r \in [R] ~ & \norm{\spc{\tv}{j}_r-\abs{\spc{\Lambda}{j}(r)}\spc{\mu}{j}_r} < \eta' \le \frac{\eta \gamma }{16 n \ell} \quad \text{ and }\\
\forall r \in [R] ~ & \norm{\spc{\tv}{\ell}_r-\abs{\spc{\Lambda}{\ell}(r)}w_r\spc{\mu}{\ell}_r} < \eta' \le \frac{\eta \gamma }{16\ell n} \label{eq:mvparams3}
\end{align}

Moreover, the $\spc{\mu}{j}_r$ correspond to probability vectors which have $\| \spc{\mu}{j} \|_1 =1$, we have ensured that $\nrm{\spc{\tv}{j}_r}{1}\in [1-\delta,1+\delta]$. Applying Lemma~\ref{lem:l1error} we get that the required estimates $\spc{\tv}{j}_r$ (i.e. $\spc{\tm}{j}_r$) satisfy:

\[\forall j \in [\ell-1], r \in [R], ~ \norm{\spc{\tv}{j}_r - \spc{\mu}{j}_r}\le \frac{\eta \gamma}{4\ell \sqrt{n}} \quad \text{ and } \quad \abs{\spc{\Lambda}{j}(r)}\in \left[ 1- \frac{\eta \gamma}{8 \ell \sqrt{n}}, 1- \frac{\eta \gamma}{8 \ell \sqrt{n}} \right] \label{eq:mvparams2}\]
 
Now, set $\spc{\tm}{\ell}_r = \frac{\spc{\tv}{\ell}_r}{\nrm{\spc{\tv}{\ell}_r}{1}}$, and $\tw_r = \nrm{\spc{\tv}{\ell}_r}{1}$, for all $r \in [R]$. Now, from equations~\eqref{eq:mvlambda} and \eqref{eq:mvparams2} we get that 
\begin{align*}
\forall r \in [R] \quad & ~\abs{\spc{\Lambda}{\ell}(r)-1} \le \frac{\eta \gamma}{8 \sqrt{n}}\\
\text{Hence from \eqref{eq:mvparams3}, }&~\norm{\spc{\tv}{\ell}_r - w_r \spc{\mu}{\ell}_r} \le \frac{\eta \gamma}{4 \sqrt{n}}\\
&~\norm{\tw_r \spc{\tm}{\ell}_r - w_r \spc{\mu}{\ell}_r} \le \frac{\eta \gamma}{4 \sqrt{n}}\\
&w_r \norm{\frac{\tw_r}{w_r} \spc{\tm}{\ell}_r - \spc{\mu}{\ell}_r} \le \frac{\eta \gamma}{4 \sqrt{n}}\\
\end{align*}
Using the fact that $w_r \ge \gamma$ and using Lemma~\ref{lem:l1error}, we see that $\tw_r$ and $\spc{\tm}{\ell}_r$ are also $\eta$-close estimates to $w_r$ and $\spc{\mu}{\ell}_r$ respectively, for all $r$. 
\end{proof}

We will now see two popular latent variable models which fit into the multi-view mixture model: the exchangeable (single) Topic Model and Hidden Markov models. We note that the results of this section (for $\ell=3$ views) also apply to other latent variable models like \emph{Latent Dirichlet Allocation (LDA)} and \emph{Independent Component Analysis (ICA)} that were studied in \cite{AGHKT12}. Anandkumar et al. \cite{AFHKL12,AGHKT12} show how we can obtain third order tensors by looking at ``third'' moments and applying suitable transformations. Applying our robust uniqueness theorem (Theorem~\ref{thm:unique3}) to these $3$-tensors identify the parameters. We omit the details in this version of the paper.
 
\subsection{Exchangeable (single) Topic Model}

The simplest latent variable model that fits the multi-view setting is the Exchangeable Single Topic model as given in \cite{AHK12}. This is a simple bag-of-words model for documents, in which the words in a document are assumed to be exchangeable. This model can be viewed as first picking the topic $r \in [R]$ of the document, with probability $w_r$.  Given a topic $r \in [R]$, each word in the document is sampled independently at random according to the probability distribution $\mu_r \in \R^n$ ($n$ is the dictionary size). In other words, the topic $r \in R$ is a latent variable such that the $\ell$ words in a document are conditionally i.i.d given $r$.

The views in this case correspond to the words in a document. This is a special case of the multi-view model since the distribution of each of the views $j \in [\ell]$ is identical. 
As in \cite{AHK12,AGHKT12}, we will represent the $\ell$ words in a document by indicator vectors $\spc{x}{1}, \spc{x}{2}, \dots, \spc{x}{\ell} \in \{0,1\}^n$ ($\Bd =1$ here).
Hence, the $(i_1,i_2,\dots, i_\ell)$ entry of the tensor $\E{\spc{x}{1} \otimes \spc{x}{2} \otimes \dots \spc{x}{\ell}}$ corresponds to the probability that the first words is $i_1$, the second word is $i_2$, $\dots$ and the $\ell^{th}$ word is $i_\ell$. The following is a simple corollary of Theorem~\ref{thm:MM}.

\begin{corollary}[Polynomial Identifiability of Topic Model]
The following statement holds for any constant $\delta>0$. 
Suppose we are given documents generated by the topic model described above, where the topic probabilities of the $R$ topics are $\{w_r\}_{r \in [R]}$, and the probability distribution of words in a topic $r$ are given by $\mu_r \in \R^n$ (represented as a $n$-by-$R$ matrix $M$). 
If $\forall r \in [R]~ w_r > \gamma$, and if $\krk{\ka}(M)\ge k \ge 2R/ \ell+1$,\\
then there is a algorithm that given any $\eta>0$ uses $N=\spc{\vartheta_{\ref{thm:MM}}}{\ell}\left(\frac{1}{\eta},R,n,\ka,1/\gamma,1\right)$ samples, and finds with high probability $M'$ and $\{w'_r\}_{r \in [R]}$ such that
\begin{equation}
\nrm{M-M'}{F} \le \eta \quad \text{ and }\quad \forall r \in [R], ~ \abs{w_r - w'_r} < \eta 
\end{equation}
Further, this algorithm runs in time $n^{O_{\ell}(R^2 \log (\frac{1}{\eta \gamma})}\left(\frac{n\ka}{\gamma}\right)^{O(\ell)}$ time.
\end{corollary}

\subsection{Hidden Markov Models}\label{sec:hmms}
The next latent variable model that we consider are (discrete) Hidden Markov Model which is extensively used in speech recognition, image classification, bioinformatics etc. We follow the same setting as in \cite{AMR09}: there is a hidden state sequence $Z_1,Z_2,\dots,Z_m$ taking values in $[R]$, that forms a stationary Markov chain $Z_1 \rightarrow Z_2 \rightarrow \dots \rightarrow Z_m$ with transition matrix $P$ and initial distribution $w=\{w_r\}_{r \in [R]}$ (assumed to be the stationary distribution). The observation $X_t$ is from the set of discrete events\footnote{in general, we can also allow $x_t$ to be certain continuous distributions like multivariate gaussians} $\{1,2,\dots,n\}$ and it is represented by an indicator vector in $\spc{x}{t} \in \R^n$. Given the state $Z_t$ at time $t$, $X_t$ (and hence $\spc{x}{t}$) is conditionally independent of all other observations and states. The matrix $M$ (of size $n \times R$) represents the probability distribution for the observations: the $r^{th}$ column $M_r$ represents the probability distribution conditioned on the state $Z_t=r$ i.e.
$$\forall r \in [R], \forall j \in [n], \quad \Prb{X_j = i \vert Z_j = r}= M_{ir}.$$
The HMM model described above is shown in Fig.~\ref{fig:hmm1}.

\begin{figure}
\centering
\begin{minipage}{.5\textwidth}
  \centering
  \includegraphics[width=0.7\linewidth]{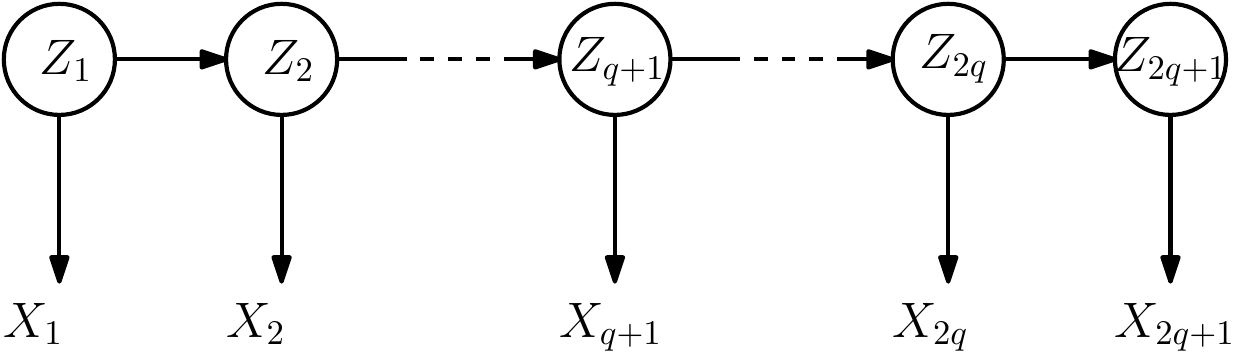}
  \caption{An HMM with $2q+1$ time steps. }
  \label{fig:hmm1}
\end{minipage}%
\begin{minipage}{.5\textwidth}
  \centering
  \includegraphics[width=\linewidth]{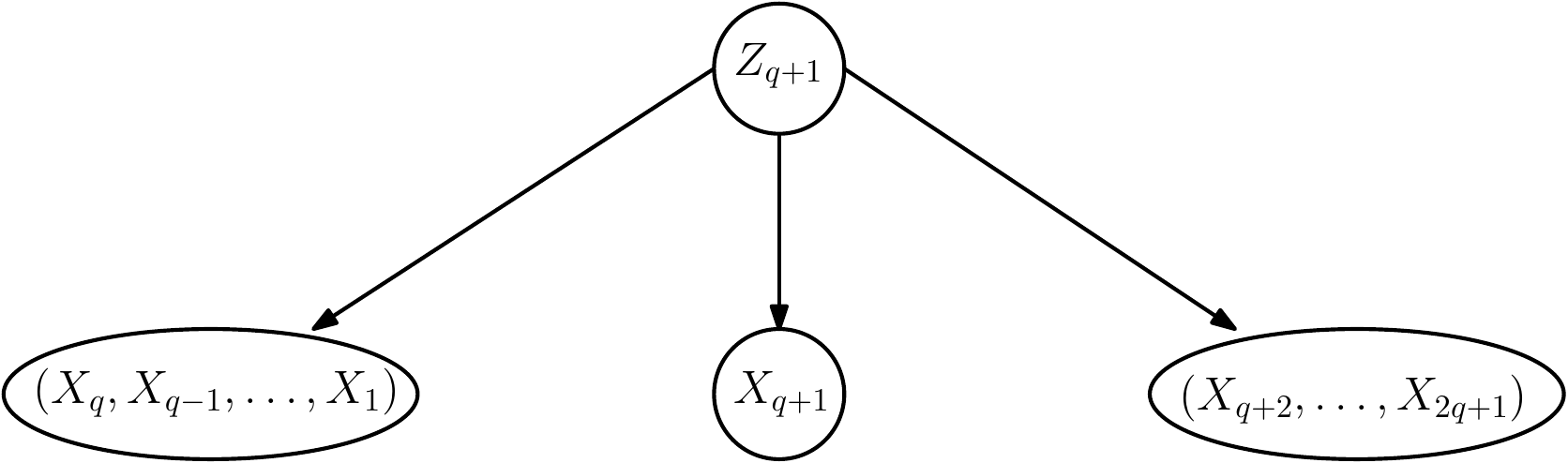}
  \caption{Embedding the HMM into the Multi-view model}
  \label{fig:hmm2}
\end{minipage}
\end{figure}

\begin{corollary}[Polynomial Identifiability of Hidden Markov models]\label{corr:hmm}
The following statement holds for any constant $\delta>0$. 
Suppose we are given a Hidden Markov model as described above, with parameters satisfying :
\begin{enumerate}[(a)]
\item The stationary distribution $\{w_r\}_{r \in [R]}$ has $\forall r \in [R]~ w_r > \gamma_1$,
\item The observation matrix $M$ has $\krk{\ka}(M)\ge k \ge \delta R$,
\item The transition matrix $P$ has minimum singular value $\sigma_{R}(P) \ge \gamma_2$,
\end{enumerate}
then there is a algorithm that given any $\eta>0$ uses $N=\spc{\vartheta_{\ref{thm:MM}}}{\frac{1}{\delta}+1}\left(\frac{1}{\eta},R,n,\ka,\frac{1}{\gamma_1 \gamma_2} \right)$ samples \\
of $m=2\lceil \frac{1}{\delta} \rceil +3$ consecutive observations (of the Markov Chain), 
and finds with high probability, $P',M'$ and $\{\tw_r\}_{r \in [R]}$ such that
\begin{equation}
\nrm{M-M'}{F} \le \eta, \quad \nrm{P-P'}{F} \le \eta ~ \text{ and } \forall r \in [R], ~ \abs{w_r - \tw_r} < \eta 
\end{equation}
Further, this algorithm runs in time $n^{O_{\delta}(R^2 \log (\frac{1}{\eta \gamma_1}))}\left(n \cdot \frac{\ka}{\gamma_1 \gamma_2}\right)^{O_\delta(1)}$ time.
\end{corollary}

\begin{proof}[Proof sketch]
The proof follows along the lines of Allman et al \cite{AMR09}, so we only sketch the proof here. 
We now show to cast this HMM into a multi-view model (Def.~\ref{def:MM}) using a nice trick of \cite{AMR09}. We can then apply Theorem~\ref{thm:MM} and prove identifiability (Corollary~\ref{corr:hmm}). We will choose $m=2q+1$ where $q=\lceil \frac{1}{\delta} \rceil+1$, and then use the hidden state $Z_{q+1}$ as the latent variable $h$ of the Multi-view model.
We will use three different views ($\ell=3$) as shown in Fig.~\ref{fig:hmm2}: the first view $A$ comprises the tuple of observations $(X_q,X_{q-1},\dots,X_1)$ (ordered this way for convenience), the second view $B$ is the observation $X_{q+1}$, while the third view $C$ comprises the tuple $V_3=(X_{q+2},X_{q+3},\dots,X_{2q+1})$. This fits into the Multi-view model since the three views are conditionally independent given the latent variable $h=Z_{q+1}$. 

Abusing notation a little, let $A, B, C$ be matrices of dimensions $n^q \times R, n \times R, n^q \times R$ respectively. They denote the conditional probability distributions as in Definition~\ref{def:MM}. For convenience, let $\tP = \diag(w) P^T \diag(w)^{-1}$, which is the ``reverse transition'' matrix of the Markov chain given by $P$. We can now write the matrices $A, B, C$ in terms of $M$ and the transition matrices. The matrix product $X \odot Y$ refers to the Khatri-Rao product (Lemma~\ref{lem:krprod}). Showing that these are indeed the transition matrices is fairly straightforward, and we refer to Allman et al.~\cite{AMR09} for the details. 

\begin{align} \label{eq:hmm}
A &= ((\dots (M\tP) \odot M)\tP) \odot M) \dots \tP) \odot M) \tP \\
B &= M \\
C &= ((\dots (MP) \odot M)P) \odot M) \dots P) \odot M) P
\end{align}

(There are precisely $q$ occurrences of $M, P$ (or $\tP$) in the first and third equalities). Now we can use the properties of the Khatri-Rao product. For convenience, define $C^{(1)}=MP$, and $C^{(j)} = (C^{(j-1)} \odot M)P$ for $j \ge 2$, so that we have $C = C^{(q)}$. By hypothesis, we have $\krk{\ka}(M) \ge k$, and thus $\krk{\ka_2 \ka}(MP) \ge k$ (because $P$ is a stochastic matrix with all eigenvalues $\ge \ka_2$). Now by the property of the Khatri-Rao product (Lemma~\ref{lem:krprod}), we have $\krk{(\ka \ka_2)\ka}(C^{(2)}) \ge \min\{R, 2k\}$. We can continue this argument, to eventually conclude that $\krk{\ka'}(C^{(q)}) = \min\{R, qk \}=R$ for $\ka'=\ka^q \gamma_2^{q^2}(qk)^{q/2}$.

Precisely the same argument lets us conclude that $\krk{\ka'}(A) \ge R$, for the $\ka'=\ka^q \gamma_2^{q^2}(qk)^{q/2}$. Now since $\krk{\ka}(B)\ge 2$, we have that the conditions of Theorem~\ref{thm:unique3} hold. Now using the arguments of Theorem~\ref{thm:MM} (here, we use Theorem~\ref{thm:unique3} instead of Theorem~\ref{thm:unique:gen}), we get matrices $A',B',C'$ and weights $w'$ such that
\begin{align*}
\nrm{A'-A}{F} &< \delta \quad \text{ and similarly for } B,C \\
\norm{w'-w} &< \delta 
\end{align*}
for some $\delta = \poly(1/\eta, \dots)$. Note that $M=B$. We now need to argue that we can obtain a good estimate $P'$ for $P$, from $A',B',C'$. This is done in \cite{AMR09} by a trick which is similar in spirit to Lemma~\ref{lem:tensoring}. It  uses the property that the matrix $C$ above is full rank (in fact well conditioned, as we saw above), and the fact that the columns of $M$ are all probability distributions.

Let $D = C^{(q-1)}$, as defined above. Hence, $C= (D \odot M)P$. Now note that all the columns of $M$ represent probability distributions, so they add up to $1$. Thus given $D \odot M$, we can combine (simply add) appropriate rows together to get $D$. Thus by performing this procedure (adding rows) on $C$, we obtain $DP$. Now, if we had performed the entire procedure by replacing $q$ with $(q-1)$ (we should ensure that $(q-1)k \ge R$ for the Kruskal rank condition to hold), we would obtain the matrix $D$. Now knowing $D$ and $DP$, we can recover the matrix $P$, since $D$ is well-conditioned. 
\end{proof}

\noindent\emph{ {\bf Remark:} 
Allman et al.~\cite{AMR09} show identifiability under weaker conditions than Corollary~\ref{corr:hmm} when they have infinite samples. This is because they prove their results for {\em generic} values of the parameters $M, P$ (this formally means their results hold for all $M, P$ except a set of measure zero, but they do not give an explicit characterization). Our bounds are weaker, but hold whenever the $\krk{\ka}(M) \ge \del n$ condition holds. Further, the main advantage is that our result is {\em robust} to noise: the case when we only have finite samples.}

\subsection{Mixtures of Spherical Gaussians}
\input{gaussmix.tex}

%% file: gaussmix.tex
\newcommand{\mom}{\text{Mom}}
Suppose we have a mixture of $R$ spherical gaussians in $\R^n$, with mixing weights $w_1, w_2, \dots w_R$, means $\mu_1, \mu_2, \dots, \mu_r$, and the common variance $\sigma^2$. Let us denote this mixture distribution by $\calD$, and the $n \times R$ matrix of means by $M$.

We define the $\mu$-tensor of $\ell$th order to be
\[ \mom_\ell := \sum_i w_i \mu_i^{\opr \ell}. \]

The empirical mean $\mu := \mom_1$, and can be estimated by drawing samples $x \sim D$, and computing $\E{x}$. Similarly, we will show how to compute $\mom_\ell$ for larger $\ell$ by computing higher order moment tensors, assuming we know the value of $\sigma$. We can then use the robust Kruskal's theorem (Theorem~\ref{thm:unique:gen}) and the sampling lemma (Lemma~\ref{lem:samplingerror:gaussians}) to conclude the following theorem.

\begin{theorem}\label{thm:gaussians}
Suppose we have a mixture of gaussians given by $\calD$, with hidden parameters $\{w_r\}_{r \in [R]}$ and $M$ (in particular, we assume we know $\sigma$)\footnote{As will be clear, it suffices to know it up to an inverse polynomial error, so from an algorithmic viewpoint, we can ``try all possible'' values.}. Suppose also that $\forall r \in [R]~ w_r > \gamma$, and $\krk{\ka}(M)=k$ for some $k \ge \delta R$.

Then there is a algorithm that given any $\eta>0$ and $\sigma$, uses $N=\spc{\vartheta_{\ref{thm:gaussians}}}{1/\delta}\left(\frac{1}{\eta},R,n,\ka,1/\gamma\right)$ samples drawn from $\calD$, and finds with high probability $M'$ and $\{w'_r\}_{r \in [R]}$ such that
\begin{equation}
\nrm{M-M'}{F} \le \eta \quad \text{ and }\quad \forall r \in [R], ~ \abs{w_r - w'_r} < \eta 
\end{equation}
Further, this algorithm runs in time $n^{O_{\delta}(R^2)}\left(\frac{n\ka}{\gamma}\right)^{O_\delta(1)}$ time.
\end{theorem}
\begin{proof}
This will follow the same outline as Theorem~\ref{thm:MM}. So, we sketch the proof here. 
The theorem works for error polynomial $\vartheta_{\ref{thm:gaussians}}$ being essentially as the same error polynomial in $\vartheta_{\ref{thm:MM}}$.
However, we first need to gain access to an order $\ell$-tensor, where each rank-$1$ term corresponds to a mean $\mu_r$. Hence, we show how to obtain this order-$\ell$ tensor of means, by subtracting out terms involving $\sigma$, by our estimates of moments upto $\ell$. 

Pick $\ell=\lceil \frac{2}{\delta} \rceil+2$. We will use order $\ell$ tensors given by the $\ell^{th}$ moment. 
We will first show how to obtain $\mom_\ell$, from which we learn the parameters.
The computation of $\mom_\ell$ will be done inductively. Note that $\mom_1$ is simply $\E{x}$. Now observe that
\begin{align*}
\E{x^{\opr 2}} = \E{x \opr x} &= \E{\sum_i w_i (\mu_i + \eps_i) \opr (\mu_i + \eps_i)} \\
&= \E{\sum_i w_i \mu_i \opr \mu_i} + \E{\sum_i w_i \eps_i \opr \eps_i} \\
&= \mom_2 + \sigma^2 I.
\end{align*}

We compute $\E{x^{\opr 2}}$ by sampling, and since we know $\sigma$, we can find $\mom_2$ up to any polynomially small error. In general, we have
\begin{align}\label{eq:calc-ml}
\E{x^{\opr \ell}} &= \sum_i w_i \E{ (\mu_i +\eps_i)^{\opr \ell}} \\
&= \sum_i w_i \sum_{x_j \in \{\mu_i, \eps_i\}} \E{ x_1 \opr x_2 \opr \dots \opr x_\ell}.
\end{align} 
The last summation has $2^\ell$ terms. One of them is $\mu_i^{\opr \ell}$, which produces $\mom_\ell$ on the RHS. The other terms have the form $x_1 \opr x_2 \opr \dots \opr x_\ell$, where some of the $x_i$ are $\mu_i$ and the rest $\eps_i$, and there is at least one $\eps_i$.

\newcommand{\mucap}{\widehat{\mu}}
If a term has $r$ terms being $\mu_i$ and $\ell-r$ being $\eps_i$, the tensor obtained is essentially a {\em permutation} of $\mucap(r, \ell) := \mu_i^{\opr r} \opr \eps_i^{\opr (\ell-r)}$. By permutation, we mean that the $(j_1, \dots, j_\ell)$th entry of the tensor would correspond to the $(j_{\pi(1)}, \dots, j_{\pi(\ell)})$th entry of $\mucap(r,\ell)$, for some permutation $\pi$. Thus we focus on showing how to evaluate the tensor $\mucap(r,\ell)$ for different $r, \ell$.

Note that if $\ell-r$ is odd, we have that $\E{\mucap(r, \ell)} =0$. This is because the odd moments of a Gaussian with mean zero, are all zero (since it is symmetric). If we have $\ell-r$ being even, we can describe the tensor $\E{\eps_i^{\ell-r}}$ explicitly as follows. Consider an index $(j_1, \dots, j_{\ell-r})$, and bucket the $j$ into groups of equal coordinates. For example for index $(1,2,3,2)$, the buckets are $\{(1),(22),(3)\}$. Now suppose the bucket sizes are $b_1, \dots, b_t$ (they add up to $\ell-r$). Then the $(j_1, \dots, j_{\ell-r})$th entry of $\eps_i^{\opr (\ell-r)}$ is precisely the product $m_{b_1} m_{b_2} \dots m_{b_t}$, where $m_s$ is the $s$th moment of the univariate Gaussian $\calN(0, \sigma^2)$.

The above describes the entries of $\E{\eps_i^{\ell-r}}$. Now $\E{\mucap(r, \ell)}$ is precisely $\mom_r \opr \E{\eps_i^{\ell-r}}$ (since the $\mu_i$ is fixed). Thus, since we have inductively computed $\mom_r$ for $r < \ell$, this gives a procedure to compute each entry of $\E{\mucap(r, \ell)}$. Thus each of the $2^\ell$ terms in the RHS of \eqref{eq:calc-ml} except $\mom_\ell$ can be calculated using this process. The LHS can be estimated to any inverse polynomial small error by sampling (Lemma~\ref{lem:samplingerror:gaussians}). Thus we can estimate $\mom_\ell$ up to a similar error.

Hence, we can use the algorithm from Section~\ref{sec:algo} and apply Corollary~\ref{corr:unique:sym} to obtain vectors $\{\tu_r\}_{r \in [R]}$ such that
$$\forall r \in [R] \norm{u_r - w^{1/\ell} \mu_r}< \eta.$$
Similarly, applying the same process with $\mom_\ell$ (the Kruskal conditions also hold for $\ell-1$) we get $\eta$-close approximations to $w_i^{1/(\ell-1)} \mu_r$. Now, we appeal to Lemma~\ref{lem:gaussians:weights} to obtain $\{w_r, \mu_r\}_{r \in [R]}.$ 
\end{proof}

\emph{Remark:} Note that the previous proof worked even when the gaussians are not spherical: they just need to have the same known covariance matrix $\Sigma$. 

The following lemma (used in the proof of Theorem~\ref{thm:gaussians}) allows us to recover the weights after obtaining estimates to $w_r^{1/\ell} \mu_r$ and $w_r^{1/(\ell-1)} \mu_r$ through decompositions for the $\ell-1$ and the $\ell$ moment tensors.
\begin{lemma}[Recovering Weights] \label{lem:gaussians:weights}
For every $\delta'>0,w>0,\minl>0,\ell \in \bbN$, $\exists \delta =\Omega\big(\frac{\delta_1 w^{1/(\ell-1)}}{\ell^2 \minl}\big)$ such that, if $\mu \in \R^n$ be a vector with length $\norm{\mu} \ge \minl$, and suppose
\[ \norm{v-w^{1/\ell} \mu} < \delta \quad \text{ and } \norm{u-w^{1/(\ell-1)} \mu} < \delta. \label{eq:weights:1}\]
Then, 
\begin{equation}
\qquad \abs{\left(\frac{\abs{\Iprod{u}{v}}}{\norm{u}}\right)^{\ell(\ell-1)} - w } < \delta'
\end{equation}
\end{lemma}
\begin{proof}
From \eqref{eq:weights:1} and triangle inequality, we see that 
$$\norm{w^{-1/\ell} v-w^{-1/(\ell-1)}u} \le \delta(w^{-1/(\ell)}+w^{-1/(\ell-1)})= \delta_1.$$
Let $\alpha_1=w^{-1/(\ell-1)}$ and $\alpha_2=w^{-1/\ell}$.
Suppose $v=\beta u + \eps \up$ where $\up$ is a unit vector perpendicular to $u$. Hence $\beta=\Iprod{v}{u}/\norm{u}$.
\begin{align*}
\norm{\alpha_1 v- \alpha_2 u}^2 &= \norm{(\beta \alpha_1-\alpha_2 )u +\alpha_1\eps \up}< \delta_1^2\\
(\beta \alpha_1 - \alpha_2)^2 \norm{u}^2 + \alpha_1^2 \eps^2 &\le \delta_1^2\\
\abs{\beta - \frac{\alpha_2}{\alpha_1}} &< \frac{\delta_1}{ \minl}
\end{align*}
Now, substituting the values for $\alpha_1,\alpha_2$, we see that
$$\abs{\beta - w^{\frac{1}{(\ell-1)} - \frac{1}{\ell}} } < \frac{\delta_1}{\minl}.$$
\begin{align*}
\abs{\beta - w^{1/(\ell(\ell-1))}} &< \frac{\delta}{w^{1/(\ell-1)} \minl}\\ 
\abs{\beta^{\ell (\ell-1)} - w } &\le \delta' \quad \text{when } \delta \ll \frac{\delta' w^{1/(\ell-1)}}{\ell^2 \minl}
\end{align*}
\end{proof}

The following Corollary establishes polynomial identifiability for mixtures of uniform spherical gaussians under milder conditions than \cite{HK12} (in particular, the means need not be in general position). The difference now is that we do not assume we know $\sigma$.

\begin{corollary}\label{corr:gaussians:unknown}
Suppose we have a mixture $\calD$ of $R$-gaussians in $n$-dimensions with $n \ge R$, with hidden parameters $\{w_r\}_{r \in [R]}$, $M$ and $\sigma$. Suppose $\forall r \in [R]~ w_r > \gamma$, and that $\krk{\ka}(M)=k$ for some $k \ge \delta R$.

Then there is a algorithm that given any $\eta>0$, uses $N=\vartheta_{\ref{thm:MM}}^{1/\delta}\left(\frac{1}{\eta},R,n,\ka,1/\gamma\right)$ samples drawn from $\calD$, and finds with high probability $\sigma'$, $M'$ and $\{w'_r\}_{r \in [R]}$ such that
\begin{equation}
\nrm{M-M'}{F} \le \eta \quad \text{ and }\quad \forall r \in [R], ~ \abs{w_r - w'_r} < \eta \quad \text{ and } \abs{\sigma - \sigma'} < \eta
\end{equation}
Further, this algorithm runs in time $n^{O_{\delta}(R^2)}\left(\frac{n\ka}{\gamma}\right)^{O_\delta(1)}$ time.
\end{corollary}
\begin{proof}[Proof sketch]
We first obtain $\sigma$ to inverse polynomial accuracy, using an elegant trick of \cite{HK13}, and then apply Theorem~\ref{thm:gaussians} to identify the parameters $M$ and weights $\{w_r\}_{r \in [R]}$. 

To estimate $\sigma$, we consider the matrix $A=\E{(x - \mom_1)\otimes (x-\mom_1)}$, and note that the estimated $n^{th}$ singular value $\sigma_n(A) \in [\sigma-\eta,\sigma-\eta]$ after averaging enough samples (see Theorem 1 in \cite{HK13} for details). This is because the $R$ vectors $\mu_i - \mom_1$ live in a $(R-1)\le n-1$ dimensional space. Hence, we can obtain $\sigma$ to any inverse polynomial accuracy (\cite{HK13} for details). This allows to recover the parameters using Theorem~\ref{thm:gaussians}. We omit the details in this version.
\end{proof}

%% file: openproblems.tex
The most natural open problem arising from our work is that of computing approximate small rank decompositions efficiently. While the problem is
NP hard in general, we suspect that {\em well conditioned} assumptions regarding robust Kruskal ranks being sufficiently large, as in the uniqueness theorem (Theorem~\ref{thm:unique3}) for decompositions of $3$-tensors for instance, could help. In particular,
\begin{question}\label{qn:algos}
Suppose $T$ is a $3$-tensor, that is promised to have a rank $R$ decomposition $\tens{A}{B}{C}$, with $k_A = \krk{\tau}(A)$ (similarly $k_B$ and $k_C$) satisfying $k_A + k_B + k_C \ge 2R+2$.
Can we find the decomposition $A, B, C$ (up to a specified error $\eps$) in time polynomial in $n, R$ and  $1/\eps$?
\end{question}
In the special case that the decomposition $\tens{A}{B}{C}$ is known to be orthogonal (i.e., the columns of $A, B, C$ are mutually orthogonal), which in particular implies
$n \ge R$, then iterative methods like power iteration \cite{AGHKT12}, and ``alternating least squares'' (ALS) \cite{CLdA09} \footnote{This is the method of choice in practice for computing tensor decompositions.} converge in polynomial time.

A result in the spirit of finding weaker sufficient conditions for uniqueness was by Chiantini and Ottaviani~\cite{COttaviani12}, who use ideas from algebraic geometry (in particular a notion called {\em weak
defectivity}), to prove that {\em generic} $n \times n \times n$ tensors of rank $k \le n^2/16$ have a unique decomposition (here the word `generic' is meant to mean all except a measure zero set of rank $k$ tensors, which they characterize in terms of weak defectivity). Note that this is much stronger
than the bound obtained by Kruskal's theorem, which is roughly $3n/2$. It is
also roughly the best one can hope for, since every $3$-tensor has rank at most $n^2$ (and a random tensor has rank $\ge n^2/2$). It would be very interesting
to prove robust versions of their results, as it would imply identifiability for a much larger range of parameters in the models we consider.

A third question is that of certifying that a given decomposition is unique.
Kruskal's rank condition, while elegant, is not known to be verifiable in
polynomial time. Given an $n \times R$ matrix, certifying that every $k$ columns are linearly independent is known to be NP-hard \cite{Khach,PT12}. 
Even the average case version i.e. when the matrix is random with independent gaussian entries, has received much attention as it is related to certifying the Restricted Isometry Property (RIP), which plays a key role in compressed sensing \cite{CTao05,KZ11}. It is thus an fascinating open question to find uniqueness (and robust uniqueness) theorems which involve parameters that can be computed efficiently.

From the perspective of learning latent variable models, it would be very interesting to obtain efficient learning algorithms with polynomial running times for the settings considered in Section~\ref{sec:applications}. Recall that we give algorithms which need only polynomial samples (in the dimension $n$, and number of mixtures $R$), when the parameters satisfy the robust Kruskal conditions. Note that an affirmative answer to Question~\ref{qn:algos} (and its higher order analogue) would already imply such efficient learning algorithms. Finally, we believe that our approach can be extended to learning the parameters of general mixtures of gaussians \cite{MV10,BS10}, mixtures of product distributions \cite{FOS05}, and more generally to a broader class of parameter learning problems. 

%% file: basiclemmas.tex
\section{A Medley of Auxiliary Lemmas}\label{sec:app:linear-algebra}
We now list some of the (primarily linear algebra) lemmas we used in our proofs. They range in difficulty from trivial to `straightforward', but we include them for completeness.

\begin{lemma}\label{lem:conditioned1}
Suppose $X$ is a matrix in $\R^{n \times k}$ with $\sigma_k \ge 1/\ka$. Then if $\norm{\sum_i \alpha_i X_i}_2 < \eps$, for some $\alpha_i$, we have $\norm{\alpha}=\sqrt{\sum_i \alpha_i^2} \le \ka \eps$.
\end{lemma}
\begin{proof}
From the singular value condition, we have for any $y \in \R^k$,
\[ \norm{Xy}_2^2 \ge \sigma_k^2 \norm{y}^2, \]
from which the lemma follows by setting $y$ to be the vector of $\alpha_i$.
\end{proof}

\begin{lemma} \label{lem:conditioned:span}
Let $A \in \R^{n \times R}$ have $\krk{\ka}=k$ and be $\ub$-bounded. Then, 
\begin{enumerate}
\item If $\calS=\spn(S)$, where $S$ is a set of at most $k-1$ column vectors of $A$, then each unit vector in $\calS$ has a small representation in terms of the columns denoted by $S$:
\[v=\sum_{i \in S} z_i A_i  \implies  \frac{1}{(\ub^2+1)k} \le (\sum_i z_i^2)/\norm{v}^2 \le \max\{\ka^2,1\} \]
\item If $\calS=\spn(S)$ where $S$ is any subset of $k-1$ column vectors $S$ of $A$, the other columns are far from the span $\calS$:
\[\forall j \in [R]\setminus S, \quad \norm{\Pi^{\perp}_{\calS} A_j} \ge \frac{1}{\ka}  \]
\item If $\calS$ is any $\ell$-dimensional space with $\ell<k$, then at most $\ell$ column vectors of $A$ are $\eps$-close to it for $\eps=1/(\ka\sqrt{\ell})$:
\[\abs{ \{i: \norm{\Pi^{\perp}_{\calS} A_i} \le \frac{1}{\ka\sqrt{\ell}}\}} \le \ell \]
\end{enumerate}
\end{lemma}
\begin{proof}
We now present the simple proofs of the three parts of the lemma.
\begin{enumerate}
\item The first part simply follows because from change of basis. Let $M$ be the $n \times n$ matrix, where the first $|S|$ columns of $M$ correspond to $S$ and the rest of the $n-|S|$ columns being unit vectors orthogonal to $\calS$. Since $A_{|S}$ is well-conditioned, then $\lmax{M} \le (\ub+1)\sqrt{n}$ and $\lmin{M} \ge 1/\max{\ka,1}$. The change of basis matrix is exactly $M^{-1}$: hence $z= (M)^{-1} v$. Thus, $\lmin{M^{-1}} \le \norm{z} \le \lmax{M^{-1}} =1/\lmin{M} \le \max\{1,\ka\}$.

\item Let $S=\{1,\dots,k-1\}$ and $j=k$ without loss of generality. Let $v=\sum_{i \in S} z_i A_i$ be a vector $\eps$-close to $A_k$. Let $M'$ be the $n \times k$ matrix restricted to first $k$ columns: i.e. $M'=A|_{S \cup \{j\}}$. Hence, the vector $z=(z_1,\dots,z_{k-1},-1)$ has square length $1+\sum_i z_i^2$, and $\norm{M'z}=\eps$. Thus, 
\[ \eps \ge \lmin{M'} \sqrt{1+\sum_i z_i^2} \ge 1/\ka \]

\item Let $\eps=1/(\ka\sqrt{k})$. For contradiction, assume that $S=\{ i: \norm{\Pi^{\perp}_{\calS} A_i} \le \eps \}$ is of size $\ell+1$. Let $v_i = \Pi_\calS A_i \in \calS$. Since $\{v_i\}_{i \in S}$ are $\ell+1$ vectors in a $\ell$ dimension space,  
$$ \exists \{\alpha_i\}_{i \in S} \text{ with } \sum_i \alpha_i^2 = 1, \quad \text{s.t } \sum_i \alpha_i v_i = 0 $$
Hence, $\norm{\sum_{i \in S} \alpha_i A_i } \le \norm{\sum_{i \in S} \alpha_i \Pi^{\perp}_{\calS} A_i } \le (\sum_{i \in S} \abs{\alpha_i}) \eps  \le \sqrt{|S|} \eps $ (where the last inequality follows from Cauchy-Schwarz inequality). But these set of $\alpha_i$ contradict the fact that the minimum singular value of any $n$-by-$k$ submatrix of $A$ is at least $1/\ka$. 
\end{enumerate}
\end{proof}

\begin{lemma}\label{lem:perm:basecase}
Let $u_1, \dots, u_t \in \R^d$ (for some $t,d$) satisfy $\norm{u_i}_2 \ge \eps >0$ for all $i$. Then there exists a unit vector $w \in \R^d$ s.t. $\abs{ \iprod{u_i, w}} > \frac{\eps}{20 dt}$ for all $i \in [t]$.
\end{lemma}
\begin{proof}
The proof is by a somewhat standard probabilistic argument.

Let $r \sim \R^d$ be a random vector drawn from a uniform spherical Gaussian with a unit variance in each direction. It is well-known that for any $y \in \R^d$, the inner product $\iprod{y, r}$ is distributed as a univariate Gaussian with mean zero, and variance $\norm{y}_2^2$. Thus for each $y$, from standard anti-concentration properties of the Gaussian, we have
\[ \Pr \big[ |\iprod{u_i, r}| \le \frac{\norm{u_i}}{10t} \big] \le \frac{1}{2t}. \]
Thus by a union bound, with probability at least $1/2$, we have
\begin{equation}\label{eq:basecase-help}
\Pr \big[ |\iprod{u_i, r}| > \frac{\eps}{10t} \big] \qquad \text{for all $i$.} \end{equation}
Next, since $\E{\norm{r}_2^2} = d$, $\Pr[\norm{r}_2^2 > 4d] < 1/4$, and thus there exists a vector $r$ s.t. $\norm{r}_2^2 \le 4d$, and Eq.~\eqref{eq:basecase-help} holds. This implies the lemma (in fact we obtain $\sqrt{d}$ in the denominator).
\end{proof}


\begin{lemma}[\kr of the Khatri-Rao product]\label{lem:krprod}
 has 
$\krk{(\ka_1 \ka_2 \sqrt{k_A+k_B})} (M) \ge \min\{k_1 + k_2 -1,R\}$. 
\end{lemma}

\begin{proof}
Let $\ka=\ka_1 \ka_2 \sqrt{k_A+k_B}$.
Suppose for contradiction $M$ has $\krk{\ka}(M)< k= k_A+k_B-1 \le R$ (otherwise we are done). \\
Without loss of generality let the sub-matrix $M'$ of size $(n_1 n_2) \times k$, formed by the first $k$ columns of $M$ have $\lambda_{k}(M) < 1/\ka$. Note that for a vector $z\in R^{nR}$, $\nrm{z}{2}=\nrm{Z}{F}$ where $Z$ is the natural $n \times R$ matrix representing $z$. Hence
$$ \exists \{\alpha_i\}_{i \in [k]} ~\text{with } \sum_{i \in [k]} \alpha_i^2 =1  \quad \text{s.t. } ~\nrm{\sum_{i \in [k]} \alpha_i A_i \otimes B_i}{F}  < \eps. \label{eq:krproduct}$$
Clearly $\exists i^* \in [k]$ s.t $|\alpha_i| \ge 1/\sqrt{k}$ : let $i^*= k$ without loss of generality.
Let $\calS=\spn(\{A_1, A_3, \dots A_{k_A-1}\})$, and pick $x = \Pi_{\calS}^{\perp} A_{k} / \norm{\Pi_{\calS}^{\perp} A_{k}}$ (it exists because $\krk{\ka}(M) <R$). \\
Pre-multiplying the expression in ~\eqref{eq:krproduct} by $x$, we get
$$\norm{\sum_{i=k_A}^{k} \beta_i B_i} < \eps ~\text{where }\beta_i = \alpha_i \Iprod{x}{A_i}$$
But $|\beta_k|\ge 1/(\sqrt{k} \ka_1)$ (by Lemma~\ref{lem:conditioned:span}), and there are only $k-k_A+1 \le k_B$ terms in the expression. Again, by Lemma~\ref{lem:conditioned:span} applied to these (at most) $k_B$ columns of $B$, we get that $1/\eps<  \ka_1 \ka_2 \sqrt{k}$, which establishes the lemma. 
\end{proof}

\noindent {\bf Remark.} Note that the bound of the lemma is tight in general. For instance, if $A$ is an $n \times 2n$ matrix s.t. the first $n$ columns correspond to one orthonormal basis, and the next $n$ columns to another (and the two bases are random, say). Then $\krk{10}(A) = n$, but for any $\tau$, we have $\krk{\tau}(A \odot A) = 2n-1$, since the first $n$ terms and the next $n$ terms of $A \odot A$ add up to the same vector (as a matrix, it is the identity).

\begin{lemma}\label{lem:tensoring}
Suppose $\nrm{u \otimes v - u' \otimes v'}{F} < \delta$, and $\minl \le \norm{u}, \norm{v}, \norm{u'}, \norm{v'} \le \maxl$,\\ with $\delta <\frac{\min\{ \minl^2 ,1 \}}{(2\max\{\maxl,1\})}$. If $u=\alpha_1 u'+\beta_1 \up$ and $v=\alpha_2 v'+\beta_2 \vp$, where $\up$ and $\vp$ are unit vectors orthogonal to $u', v'$ respectively, then we have
$$ |1-\alpha_1 \alpha_2|< \delta/ \minl^2 ~\text{ and } \quad \beta_1 < \sqrt{\delta},~\beta_2 < \sqrt{\delta}.$$
\end{lemma}
\begin{proof}
We are given that $u=\alpha_1 u' +\beta_1 \up$ and $v=\alpha_2 v' + \beta_2 \vp$. Now, since the tensored vectors are close

\begin{align}
\nrm{u \otimes v - u' \otimes v'}{F}^2 &< \delta^2 \notag\\  
\nrm{(1-\alpha_1 \alpha_2) u'\otimes v' + \beta_1 \alpha_2 \up \otimes v'+ \beta_2 \alpha_1 u' \otimes \vp + \beta_1\beta_2 \up \otimes \vp}{F}^2 &< \delta^2 \notag\\
\minl^4 (1-\alpha_1 \alpha_2)^2 + \beta_1^2 \alpha_2^2 \minl^2 +\beta_2^2 \alpha_1 ^2 \minl^2 + \beta_1^2 \beta_2^2 &< \delta^2 \label{eq:tensoring:lb}
\end{align}
This implies that $|1-\alpha_1 \alpha_2| < \delta/ \minl^2$ as required.

Now, let us assume $\beta_1> \sqrt{\delta}$. This at once implies that $\beta_2 < \sqrt{\delta}$. \\
Also 
\begin{align*}
\minl^2 \le \norm{v}^2 &= \alpha_2^2 \norm{v'}^2 + \beta_2^2  \\
\minl^2 -\delta &\le \alpha_2^2 \maxl^2 \\
\text{ Hence, }\quad \alpha_2 &\ge \frac{\minl}{2\maxl} 
\end{align*}
Now, using \eqref{eq:tensoring:lb}, we see that $\beta_1 < \sqrt{\delta}$.
\end{proof}

\begin{lemma}\label{lem:l1error}
For $\lambda \ge 0$, a vector $v \in \R^n$ with $\nrm{v}{1} \in [1-\varepsilon/4, 1+\varepsilon/4]$, a probability vector $u \in \R^n$ ( $\nrm{u}{1}=\sum_i u_i = 1$), if 
$$\nrm{v-\lambda u}{2} \le \frac{\varepsilon}{4\sqrt{n}}$$ 
then we have
\begin{align*}
1-\varepsilon/2 \le \lambda \le 1+\varepsilon/2 \quad \text{and} \quad \nrm{v-u}{2} \le \varepsilon
\end{align*}
\end{lemma}
\begin{proof}
First we have $\nrm{v-\lambda u}{1}\le \varepsilon/4$ by Cauchy-Schwartz. 
Hence, by triangle inequality, $|\lambda| \nrm{u}{1} \le 1+ \varepsilon/2$. \\
Since $\nrm{u}{1}=1$, we get $\lambda \le 1+\varepsilon/2$. Similarly $\lambda \ge 1-\varepsilon/2$.

Finally, $\nrm{v-u}{2} \le \nrm{v-\lambda u }{2}+\abs{\lambda-1} \nrm{u}{2} \le \varepsilon$ (since $\lambda \ge 0$). 
Hence, the lemma follows. 
\end{proof}

\subsection{Symmetric Decompositions} \label{sec:unique:sym}
\begin{proof}[Proof of Corollary~\ref{corr:unique:sym}]
Applying Theorem~\ref{thm:unique:gen} with $\eps'< \eta (2\rho \ka \sqrt{R})^{-1}$, to obtain a permutation matrix $\Pi$ and scalar matrices $\Lambda_j$ such that
\begin{align*}
\forall j \in [\ell] \nrm{V - U \Pi \Lambda_j}{F}<\eps'\\
\text{By triangle inequality,} \quad \forall j,j'\in [\ell], ~ \nrm{U \Pi (\Lambda_j - \Lambda_{j'})}{F}< 2\eps'
\end{align*}

Since $\Pi$ is a permutation matrix and $U$ has columns of length at least $1/\ka$, 
we get that 
\begin{equation*}\label{eq:sym:1}
\forall r \in [R], j \in [\ell], j' \in [\ell], \quad \abs{\Lambda_j (r)-\Lambda_{j'}(r)} < \eps' \ka
\end{equation*}
However, we also know that 
\begin{align*}
\norm{\prod_{j \in \ell}, ~ \Lambda_j -I} &\le \eps'\\
\forall r \in [R], \quad (1-\eps')&\le \prod_{j \in [\ell]} \Lambda_j(i) \le 1+\eps'
\end{align*}
Hence, substituting \eqref{eq:sym:1} in the last inequality, it is easy to see that 
$\forall i \in [n], \abs{\lambda_j (i)-1 } < 2 \eps' \ka$. 
But since each column of $A$ is $\rho$-bounded, this shows that $\nrm{A'-A \Pi}{F}<2 \eps' \ka \rho \sqrt{R} \le \eta $, as required.
\end{proof}

\section{Properties of Tensors} \label{sec:badexamples}

\subsection{A necessary condition for Uniqueness} \label{sec:krrank}

Consider a $3$-tensor $T$ of rank $R$ represented by $\tens{A}{B}{C}$ where these three matrices are of size $n\times R$. 
$$ T= \sum_{r \in [R]} A_r \otimes B_r \otimes C_r.$$ 
We now show a necessary condition in terms of the $n^2$ dimensional vectors $A_r \otimes B_r$ from the decomposition.
\begin{claim}[A necessary condition for uniqueness]
Suppose for a subset $S \subset [R]$, there exist $\{\alpha_r\}$ with $\norm{\alpha}=1$.
$$ \sum_{r \in S} \alpha_r A_r \otimes B_r = 0$$
then there exists multiple rank-$R$ decompositions for $T$
\end{claim}
\begin{proof}
Consider any fixed non-zero vector $u$ (it can be also chosen to be not close to any of the other vectors in $S$).
This is because $\sum_{r \in S} A_r \otimes B_r \otimes u = \sum_{r \in S} \alpha_r (A_r \otimes B_r)\otimes u =0$.\\
Hence, $T= \sum_{r \in S} A_r \otimes B_r \otimes (C_r + \alpha_r u) + \sum_{r' \in [R]\setminus S} A_{r'} \otimes B_{r'} \otimes C_{r'}.$ 
\end{proof}

The above example showed that one necessary condition is that the $A \odot B$ should be full rank $R$ (and well-conditioned). These examples are ruled out when the Kruskal ranks of $A$ and $B$ are such that $k_A + k_B \ge R$ by Lemma~\ref{lem:krprod}.

\section{Sampling Error Estimates for Higher Moment Tensors}

In this section, we show error estimates for $\ell$-order tensors obtained by looking at the $\ell^{th}$ moment of various hidden variable models. In most of these models, the sample is generated from mixture of $R$ distributions $\{\calD_r\}_{r \in [R]}$, with mixing probabilities $\{w_r\}_{r \in [R]}$. First the distribution $\calD_r$ is picked with probability $w_r$, and then the data is sampled according to $\calD_i$, which is characteristic to the application.

\begin{lemma}[Error estimates for Multiview mixture model]\label{lem:samplingerror:topic}
For every $\ell \in \bbN$, suppose we have a multi-view model, with parameters $\{w_r\}_{r \in [R]}$ and $\{\spc{M}{j}\}_{j \in [\ell]}$, such that every entry of $\spc{x}{j} \in \R^n$ is bounded by $\Bd$ (or if it is multivariate gaussian). Then,
for every $\eps>0$, there exists $N =O(\Bd^\ell \eps^{-2}\sqrt{\ell \log n})$ such that \\
if $N$ samples $\{\spc{x(1)}{j}\}_{j \in [\ell]},\{\spc{x(2)}{j}\}_{j \in [\ell]},\dots,\{\spc{x(N)}{j}\}_{j \in [\ell]}$ are generated, then with high probability
\begin{equation}\label{eq:sampleerror:topic}
\nrm{\E{\spc{x}{1} \otimes \spc{x}{2} \otimes \dots \spc{x}{\ell} }- \frac{1}{N} \left(\sum_{t \in [N]} \spc{x(t)}{1} \otimes \spc{x(t)}{2} \otimes \spc{x(t)}{\ell} \right)}{\infty} < \eps 
\end{equation}
\end{lemma}
\begin{proof}
We first bound the $\| \cdot\|_{\infty}$ norm of the difference of tensors i.e. we show that
$$\forall \{i_1,i_2,\dots,i_\ell\}\in [n]^\ell, \abs{\E{\prod_{j \in [\ell]} \spc{x}{j}_{i_j} }-\frac{1}{N}\left(\sum_{t \in [N]} \prod_{j \in [\ell]} \spc{x(t)}{j}_{i_j}\right)} < \eps/n^{\ell/2}.$$
Consider a fixed entry $(i_1,i_2, \dots, i_\ell)$ of the tensor. 

Each sample $t \in [N]$ corresponds to an independent random variable with a bound of $\Bd^\ell$. Hence, we have a sum of $N$ bounded random variables. By Bernstein bounds, probability for \eqref{eq:sampleerror:topic} to not occur $\exp\left(-\frac{\left(\eps n^{-\ell/2}\right)^2 N^2}{2N\Bd^\ell}\right)=\exp\left(-\eps^2 N/\left(2(\Bd n)^\ell\right)\right)$. We have $n^\ell$ events to union bound over. Hence $N=O(\eps^{-2} (\Bd n)^\ell \sqrt{\ell \log n})$ suffices. Note that similar bounds hold when the $\spc{x}{j} \in \R^n$ are generated from a multivariate gaussian.
\end{proof}

\begin{lemma}[Error estimates for Gaussians]\label{lem:samplingerror:gaussians}
Suppose $x$ is generated from a mixture of $R$-gaussians with means $\{\mu_r\}_{r \in [R]}$ and covariance $\sigma^2 I$ , with the means satisfying $\norm{\mu_r} \le \mmax$.\\
For every $\eps>0, \ell \in \bbN$, there exists $N =\Omega( \poly(\frac{1}{\eps})), \sigma^2, n,R)$ such that 
if $x^{(1)},x^{(2)},\dots,x^{(N)} \in R^n$ were the $N$ samples, then
\begin{equation}\label{eq:sample:errorbound}
\forall \{i_1,i_2,\dots,i_\ell\}\in [n]^\ell, \abs{\E{\prod_{j \in [\ell]} x_{i_j} }-\frac{1}{N}\left(\sum_{t \in [N]} \prod_{j \in [\ell]} x^{(t)}_{i_j}\right)} < \eps.
\end{equation}
In other words,
\[ \nrm{\E{x^{\otimes \ell} }- \frac{1}{N} \big(\sum_{t \in [N]} (x^{(t)})^{\otimes \ell}\big)}{\infty} < \eps  \]
\end{lemma}
\begin{proof}
Fix an element $(i_1,i_2,\dots,i_\ell)$ of the $\ell$-order tensor. Each point $t \in [N]$ corresponds to an i.i.d random variable $Z^t=x^{(t)}_{i_1} x^{(t)}_{i_2} \dots x^{(t)}_{\ell}$. We are interested in the deviation of the sum $S=\frac{1}{N}\sum_{t \in [N]} Z^t$. Each of the i.i.d rvs has value $Z=x_{i_1} x_{i_2} \dots x_{\ell}$. Since the gaussians are spherical (axis-aligned suffices) and each mean is bounded by $\mmax$, $|Z|< (\mmax+t \sigma)^\ell$ with probability $O\left(\exp(-t^2/2)\right)$. Hence, by using standard sub-gaussian tail inequalities, we get
$$\Pr{\abs{S-\E{z}} > \eps} < \exp\left(-\frac{\eps^2 N}{(M+\sigma \ell \log n)^\ell}\right)$$
Hence, to union bound over all $n^{\ell}$ events $N=O\left(\eps^{-2} (\ell \log n M)^\ell\right)$ suffices.
\end{proof}